\documentclass{article}

\usepackage{arxiv}
\usepackage[utf8]{inputenc} 
\usepackage[T1]{fontenc}    
\usepackage{hyperref}       
\usepackage{url}            
\usepackage{booktabs}       
\usepackage{amsfonts}       
\usepackage{nicefrac}       
\usepackage{microtype}      
\usepackage{graphicx}
\usepackage[authoryear]{natbib}
\usepackage{doi}
\usepackage{amsmath}
\usepackage{amssymb}
\usepackage{cleveref}
\usepackage{longtable}
\usepackage{array}
\usepackage{calc}
\usepackage[normalem]{ulem}
\usepackage[ruled,linesnumbered]{algorithm2e}
\usepackage{setspace}

\newtheorem{theorem}{Theorem}

\newtheorem{remark}{Remark}

\let\ul\uline

\title{Constructing confidence intervals for constrained parameters via valid prior-free inferential models}
\author{{\includegraphics[scale=0.06]{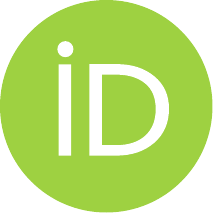}\hspace{1mm}}Hezhi Lu\thanks{Corresponding Author: Hezhi Lu; email: luhz@gzhu.edu.cn} \\
    School of Economics and Statistics \\
    Lingnan Research Academy of Statistical Science \\
    Guangzhou University \\
    Guangzhou 510006, P. R. China
    \And
    Qijun Wu \\
    School of Economics and Statistics \\
    Guangzhou University \\
    Guangzhou 510006, P. R. China
}

\renewcommand{\shorttitle}{Constructing confidence intervals for constrained parameters}

\date{}

\hypersetup{
pdftitle={Constructing confidence intervals for constrained parameters via valid prior-free inferential models},
pdfauthor={Hezhi Lu, Qijun Wu},
pdfkeywords={Constrained parameter, Confidence interval, Improper prior, Inferential model},
}

\begin{document}
\maketitle

\begin{abstract}
Constructing valid inferential methods for constrained parameters in
normal and Poisson distributions represents two fundamental and
important problems in applied statistics, for which there is currently
no unified framework for statistical inference. Most existing studies
assume that the nuisance parameters of the model are known, an
assumption that is often impractical in real-world applications.
However, under the more realistic scenario where nuisance parameters are
unknown, the available Bayesian interval estimation methods fail to
guarantee nominal coverage and thus cannot provide exact inference. To
address these limitations, this paper develops prior-free inferential
model (IM) approaches for parameters of interest in constrained normal
and Poisson models and demonstrates that the confidence intervals (CIs)
obtained from these novel IM methods can achieve exact nominal
coverage. Furthermore, considering the discrete nature of the Poisson
distribution, we employ random weighting techniques to improve the
conservative coverage performance of the IM CIs. Simulation studies show
that the coverage probabilities of the improved nonrandomized
inferential model (NIM) CIs are closest to the prespecified nominal
levels, with corresponding expected lengths shorter than those of
Bayesian intervals in weak signal scenarios, whereas the shorter
expected lengths of Bayesian intervals in strong signal scenarios come
at the cost of sacrificing coverage guarantees. Therefore, the proposed
IM and NIM CIs are superior to the Bayesian CIs. Finally, the
advantages of the proposed methods are confirmed through an analysis of
two experimental datasets on neutrinos in high-energy physics.
\end{abstract}



\section{Introduction}

In statistical inference, effectively quantifying uncertainty and
constructing interval estimates for parameters with known constraints
such as nonnegativity or boundedness remains a persistent challenge for
both frequentist and Bayesian methods. These problems are particularly
common in scientific fields such as high-energy physics, astronomical
observation, and environmental monitoring. For example, two statistical
inference problems of particular interest to high-energy physicists over
the past two decades are: inference about a nonnegative mean subject to
Gaussian measurement error and inference about a Poisson signal rate in
the presence of background contamination. In this paper, we aim to
develop more efficient inference methods for such constrained parameter
problems, which are essentially based on an inferential model (IM)
framework.

Numerous researchers have investigated constrained parameters under two
distinct modeling assumptions. Under the first assumption, for the
normal case, suppose that \(X\) is the measurement of a nonnegative
quantity, \(\theta\), with a Gaussian error distribution. For
simplicity, the variance is set to \(\sigma^{2} = 1\), which yields the
probability model \(X\sim N(\theta,1)\) subject to the constraint
\(\theta \geq 0\). For the Poisson case, let \(S\sim Poisson(\lambda)\)
denote the number of signal events and \(B\sim Poisson(\varepsilon)\)
denote the number of background events, with a known background rate
\(\varepsilon\) and an unknown signal rate \(\lambda\). Assuming the
independence of \(S\) and \(B\), only their sum \(X = B + S\) is
observed, which follows a Poisson distribution with mean
\(\theta = \varepsilon + \lambda\). Inference on \(\lambda\) is
therefore equivalent to inference on \(\theta\), subject to the physical
constraint \(\theta \geq \varepsilon\). When an inference is performed
on a constrained parameter, traditional Neyman's confidence intervals
(CIs) can yield empty intervals or exhibit unsatisfactory coverage
properties when the data are near or beyond the constraint boundaries.
\citet{Feldman1998} and \citet{Mandelkern2000} proposed unified
approach-based CIs for bounded parameters. \citet{Giunti1999}
demonstrated that the upper bound of the unified CI decreases as
\(\varepsilon\) increases, and introduced a new ordering principle to
modify the unified approach. Modest refinements have been proposed in
\citet{Mandelkern2002}, \citet{Fraser2004}, and \citet{Zhu2007}.
However, \citet{Lu2023} reported that while these frequentist CIs can
ensure a predetermined coverage probability, their actual coverage
probabilities substantially exceed the nominal coverage level.

Bayesian inference provides a more credible interval estimation method
than does the unified interval. However, this construction method yields
unphysically short confidence intervals near the boundary values of the
constraint conditions. For example, \citet{Roe2000} adopted a Bayesian
approach based on a uniform improper prior to constructing a credible
belt for constrained parameters. \citet{Zhang2003} reported that
Bayesian CIs based on a uniform prior still maintain a high coverage
probability. In fact, Bayesian methods that use noninformative priors typically lack a
clear frequentist interpretation and are sensitive to the choice of
prior specifications. To eliminate the influence of prior distributions
on inferential results, \citet{Leaf2012} proposed the elastic belief
(EB) method for constrained parameters and developed a postdata
predictive measure of uncertainty for unknown parameters. Furthermore,
\citet{Lu2023} reported that EB CIs fail to improve the conservative
performance of frequentist and Bayesian confidence intervals;
accordingly, they proposed an IM method based on random weighting,
which can enhance the coverage performance of existing CIs to a certain
extent. Therefore, the optimal CIs of constrained parameters are still
in progress.

Many research findings have been reported regarding the assumption of
the first model; however, \citet{Zhang2003} pointed out that the true
values of the nuisance parameters \(\sigma^{2}\) and \(\varepsilon\) in
normal and Poisson models are difficult to determine in practical
scenarios. They thereby proposed a more general theoretical model: In
the normal case, suppose that independent random variables
\(X\sim N(\theta,\sigma^{2})\) and \(W\sim\sigma^{2}\chi_{r}^{2}\) are
observed, where \(\theta \geq 0\) and \(\sigma > 0\) are unknown
parameters and \(r\) is a positive integer; in the Poisson case,
suppose that we observe \(X = B + S\) and \(W\), where \(B\), \(S\),
and \(W\) are independent random variables for which
\(B\sim Poisson(\varepsilon)\), \(S\sim Poisson(\lambda)\), and
\(W\sim Poisson(m\varepsilon)\), where \(\lambda,\varepsilon \geq 0\)
are unknown parameters and \(m > 0\) is known. The goal is to develop
interval estimation methods for the constrained parameters of interest,
\(\theta\) and \(\lambda\), in the context of normal and Poisson
distribution cases when the nuisance parameters \(\sigma^{2}\) and
\(\varepsilon\) are unknown.

The literature on the assumption of the second model is relatively
scarce. On the basis of our experience, only \citet{Zhang2003}
systematically investigated the construction of credible sets for
parameters under constrained parameter spaces via improper priors. Our
simulation studies reveal that Bayesian credible intervals for the
constrained parameter are unstable and fail to guarantee the nominal
coverage probability. Since the IM framework is highly suited to
scenarios where parameter spaces are subject to constraints, it can
directly integrate the constraint set of parameters with predictive
random sets, thereby addressing the challenges that constrained
parameter spaces pose to frequentist and Bayesian inferential methods.
More importantly, the IM is grounded in Dempster--Shafer
\citep{Dempster2008,Shafer1976} belief function theory: it generates
statistical inference results with frequentist calibration properties
without relying on prior information. Therefore, this paper focuses on
the second theoretical model, which is more aligned with practical
application scenarios, and aims to propose more robust and interpretable
IM-based CIs for constrained parameters. The remainder of the paper is
structured as follows: Section 2 reviews the models and Bayesian
methods from \citet{Zhang2003}. Section 3 introduces the basic framework
of the IM and proposes our IM-based inference methods for constrained
parameters. In Section 4, we compare the coverage performance of the new
method with that of Bayesian approaches through simulation studies. Two
real data examples are presented for illustration in Section 5.
Finally, we conclude with a discussion in the last section.

\section{Existing confidence intervals}

\subsection{Inference for nonnegative quantities with Gaussian errors}

Suppose that \(X\) is the measurement of a nonnegative quantity
\(\theta\), following a Gaussian error distribution
\(N(\theta,\sigma^{2})\) with mean \(\theta\) and variance
\(\sigma^{2}\), where \(\theta\) satisfies the constraint
\(C = \{\theta:\theta \geq 0\}\) and \(\sigma^{2}\) is a nuisance
parameter. In the study of \cite{Zhang2003}, the model further
assumes that the observed random variables
\(X\sim N(\theta,\sigma^{2})\) and
\(W\sim\sigma^{2}\chi_{r}^{2}\)\hspace{0pt} are mutually independent,
where \(\theta \geq 0\) and \(\sigma > 0\) are unknown parameters and
where \(r\) is a positive integer. The goal is to construct a
\(1 - \alpha\) CI for the parameter of interest \(\theta\).

The joint probability mass function of \(W\) and \(X\) is

\[
f\left( w,x|\theta,\sigma \right) = \frac{1}{2^{\frac{r + 1}{2}}\sigma^{r + 1}\sqrt{\pi}\Gamma(r/2)}w^{\frac{r}{2} - 1}\exp\left\lbrack - \frac{(x - \theta)^{2} + w}{2\sigma^{2}} \right\rbrack.
\]

Let \(S^{2} = W/r\) and \(H_{r}( \cdot )\) be the \emph{t}
distribution functions with \(r\) degrees of freedom. Suppose that
\(\theta\) and \(\sigma\) are given the improper prior
\(f(\theta,\sigma) = \frac{1}{\sigma}\) for \(0 < \sigma < \infty\) and
\(0 \leq \theta < \infty\); then, the marginal density of \(X\) and
\(W\) is

\[
f(w,x) = \int_{0}^{\infty}{\int_{0}^{\infty}{\frac{f\left( w,x|\theta,\sigma \right)}{\sigma}d\sigma d\theta}} = \frac{1}{2w}H_{r}(t),
\]

where \(t = x/s\) and \(s^{2} = w/r\). Hence, the posterior density of
\(\theta\) is

\[
g\left( \theta|w,x \right) = \int_{0}^{\infty}{\frac{f\left( w,x|\theta,\sigma \right)}{f(w,x)\sigma}d\sigma = \frac{1}{sH_{r}(t)}h_{r}\left( \frac{\theta - x}{s} \right)},
\]

where \(h_{r}( \cdot )\) is the density of \(H_{r}( \cdot )\).

A level \(1 - \alpha\) Bayesian credible interval for \(\theta\)
requires two functions \(l = l(w,x)\) and \(u = u(w,x)\), for which

\[
1 - \alpha = \int_{l}^{u}{g\left( \theta|w,x \right)d\theta},
\]

and to minimize the length of the interval \(l\) and \(u\) must satisfy

\[
\lbrack l,u\rbrack = \left\{ \theta:g\left( \theta \middle| w,x \right) \geq c \right\},
\]

for some constant \(c\). Hence, \(l = max\{ 0,x - bs\}\) and
\(u = x + bs\), where

\[
b = \max\left\{ H_{r}^{- 1}\left( 1 - \alpha H_{r}(t) \right),H_{r}^{- 1}\left( \frac{1}{2} + \frac{1}{2}(1 - \alpha)H_{r}(t) \right) \right\},
\]

and \(t = x/s\).

\textbf{2.2 Inference about the Poisson rate from a contaminated
observed count}

Suppose we observe \(X = B + S\) and \(W\), where \(B\), \(S\), and
\(W\) are mutually independent random variables such that
\(B\sim Poisson(\varepsilon)\), \(S\sim Poisson(\lambda)\), and
\(W\sim Poisson(m\varepsilon)\). Here, \(\lambda \geq 0\) and
\(\varepsilon \geq 0\) are unknown parameters, whereas \(m > 0\) is a
known constant. The goal is to infer \(\lambda\).

The joint probability mass function of \(W\) and \(X\) is

\[
p\left( w,x|\varepsilon,\lambda \right) = \sum_{k = 0}^{x}{\frac{1}{w!k!(x - k)!}m^{w}}\varepsilon^{w + k}\lambda^{x - k}e^{- \lambda - (m + 1)\varepsilon}.
\]

\cite{Zhang2003} proposed a Bayesian solution for inferring
\(\lambda\). Suppose that \(\lambda\) and \(\varepsilon\) are given an
improper prior density, e.g., \(\varepsilon^{a - 1}e^{- b\varepsilon}\),
for \(0 \leq \varepsilon,\lambda < \infty\) where \(a > 0\),
\(b \geq 0\). Then the marginal probability mass function of \(W\) and
\(X\) is

\[
p(w,x) = \int_{0}^{\infty}{\int_{0}^{\infty}{p\left( w,x|\varepsilon,\lambda \right)}}\varepsilon^{a - 1}e^{- b\varepsilon}d\varepsilon d\lambda = \sum_{k = 0}^{x}{\frac{\Gamma(a + w + k)}{w!k!}\frac{m^{w}}{{(b + m + 1)}^{a + w + k}}.}
\]

Using Bayes' theorem, the posterior density and distribution function of
\(\lambda\) are

\[
g\left( \lambda \middle| w,x \right) = \frac{1}{p(w,x)}\sum_{k = 0}^{x}{\frac{\Gamma(a + w + k)}{w!k!(x - k)!}\frac{m^{w}}{(b + m + 1)^{a + w + k}}\varepsilon^{- \lambda}\lambda^{x - k}},
\]

and

\[
G\left( \lambda \middle| w,x \right) = \frac{1}{p(w,x)}\sum_{k = 0}^{x}{\frac{\Gamma(a + w + k)}{w!k!}\frac{m^{w}}{(b + m + 1)^{a + w + k}}H_{x - k + 1}(\lambda)},
\]

where \(H_{j}( \cdot )\) denotes the gamma distribution function with
shape parameter \(j\) and unit scale parameter. Hence, a Bayesian
credible interval consists of two functions, \(l = l(w,x)\) and
\(u = u(w,x)\), for which

\[
1 - \alpha = G\left( u \middle| w,x \right) - G\left( l \middle| w,x \right),
\]

\[
\lbrack l,u\rbrack = \left\{ \lambda:g\left( \lambda \middle| w,x \right) \geq c \right\},
\]

for some \(c = c(w,x)\).

\section{New IM-based confidence intervals}

In this section, we first briefly introduce the IM framework aimed at
establishing valid probabilistic inference and then propose new IM-based
CIs for the parameter of interest \(\lambda\).

\subsection{IM framework for valid probabilistic inference}

If \(X\), denote the observable sample data. The sampling model is then
defined as a probability distribution \(P_{X|\theta}\) over the sample
space \(\mathbf{X}\), with \(\theta \in \mathbf{\Theta}\) serving as the
indexing parameter. For a given \(\theta\), this sampling model for
\(X\) is induced via an auxiliary variable \(U\). Specifically, let
\(\mathbf{U}\) be an auxiliary space endowed with a probability measure
\(P_{U}\); the sample model \(P_{X|\theta}\) is determined via the
following procedure:

sample \(U\sim P_{U}\) and set \(X = \varphi(\theta,U)\),

where \(\varphi:\mathbf{\Theta \times U \rightarrow X}\) denotes an
appropriate mapping. Furthermore, if \(X = x\) is observed, then we know
that \(x = \varphi(\theta,u^{*})\), where \(u^{*}\) is some unobserved
realization of \(U\). Clearly, knowing \(u^{*}\) is equivalent to
knowing \(\theta\). Hence, the IM approach attempts to accurately
predict the value of \(u^{*}\) before conditioning on \(X = x\).

\emph{\textbf{Definition 1} \citep{Martin2013}}. A predictive random
set \(S\) is valid for predicting the unobserved auxiliary variable if
\(Q_{S}(U) = P_{S}\left\{ U \notin S \right\}\), as a function of
\({U\sim P}_{U}\), is stochastically no larger than \(Unif(0,1)\), that
is, for each \(\alpha \in (0,1)\),
\(P_{U}\left\{ Q_{S}(U) \geq 1 - \alpha \right\} \leq \alpha\). If $\leq \alpha$ 
can be replaced by $= \alpha$, then $S$ is efficient.

The IM consists of three steps:

\textbf{\emph{Association}-step:} From an appropriate mapping
\(\varphi:X = \varphi(\theta,U)\), the IM establishes an association
between the parameter \(\theta\) and each possible pair \((X,U)\). This
association yields a collection of candidate value sets defined as
\(\Theta_{X}(U) = \{\theta:X = \varphi(\theta,U)\}\).

\textbf{\emph{Prediction}-step:} Given observed data \(X = x\), let
\(\theta^{*}\) denote the true value of \(\theta\). There exists a
corresponding \(u^{*}\) such that \(x = \varphi(\theta^{*},u^{*})\).
Furthermore, a valid predictive random set \(S(u)\) is employed to
predict this true \(u^{*}\). The validity condition ensures that
\(S(u)\) will hit \(u^{*}\) with a large probability.

\textbf{\emph{Combination}-step:} The association step and prediction
step are combined to obtain a final random set of \(\theta\), that is,
\(\Theta_{x}\left( S(u) \right) = \bigcup_{u \in S(u)}^{}{\Theta_{x}(u)}\).
Then, for any assertion \(A\) regarding the parameter of interest
\(\theta\), two probability measures are computed to quantify the
evidence in \(x\) supporting \(A\): the belief function
\({bel}_{x}(A) = P_{S(u)}\{\Theta_{x}\left( S(u) \right) \subseteq A\}\)
and the plausibility function
\({pl}_{x}(A) = P_{S(u)}\{\Theta_{x}\left( S(u) \right) \nsubseteq A^{c}\}\).

Notably, \({bel}_{x}(A)\) and \({pl}_{x}(A)\) represent the minimum and
maximum probabilities supporting the truth of assertion \(A\),
respectively. In practical applications, reporting the plausibility
function is more convenient, as it can be easily utilized to construct
frequentist statistical procedures. To test assertion
\(A = \{\theta:\theta = \theta_{0}\}\), the null hypothesis
\(H_{0}:\theta = \theta_{0}\) is rejected at a significance level
\(\alpha\) if \({pl}_{x}(A) \leq \alpha\). Additionally, this
plausibility function yields a two-sided \(1 - \alpha\) IM CI
\(\{\theta:{pl}_{x}(A) > \alpha\}\).

\emph{\textbf{Definition 2} \citep{Martin2013}}. Suppose that
 $X \sim P_{X|\theta}$ and let $A$ be an assertion of interest. 
The IM with the plausibility function $\mathrm{pl}_X(A)$ is 
valid for assertion $A$ if, for each $\alpha \in (0,1)$,
\[
\sup_{\theta \in A} P_{X|\theta} \left\{ \mathrm{pl}_X(A) \le \alpha \right\} \le \alpha .
\]
The IM is valid if it is valid for all $A$.

\begin{theorem}
\label{thm:1}
Suppose that the predictive random set $S$ is valid and
 that $\Theta_x(S) \neq \varnothing$ with $P_S$-probability 1 for all $x$. 
Then the IM confidence interval $\{\theta : \mathrm{pl}_x(A) > \alpha\}$
can guarantee the nominal coverage probability. Moreover, 
if ``$\le \alpha$'' can be replaced by ``$= \alpha$'', then the IM interval controls 
the coverage probability exactly at the confidence level $1-\alpha$.

\end{theorem}

\subsection{Inference for the constrained normal case}

Suppose that independent random variables \(X\sim N(\theta,\sigma^{2})\)
and \(W\sim\sigma^{2}\chi_{r}^{2}\) are observed, where
\(\theta \geq 0\) and \(\sigma > 0\) are unknown parameters and \(r\) is
a positive integer. Let \(Z\sim N(0,1)\), \(U\sim\chi_{r}^{2}\), where
\(Z\) and \(U\) are independent unobservable but predictable auxiliary
variables. The association linking \(X\), \(\theta\), and an auxiliary
variable \(Z\) may be written as

\begin{equation}\label{eq:1}
X = \theta + \sigma Z,\ Z\sim N(0,1).
\end{equation}

Similarly, the association for \(W\), given \(\sigma^{2}\), may be
written as

\begin{equation}\label{eq:2}
W = \sigma^{2}U,U\sim\chi_{r}^{2}.
\end{equation}

From (\ref{eq:2}), \(\sigma^{2} = W/U\); we then plug \(\sigma^{2} = W/U\) into
(\ref{eq:1}), that is,

\begin{equation}\label{eq:3}
X = \theta + \sqrt{W}\frac{Z}{\sqrt{U}},\ Z\sim N(0,1),U\sim\chi_{r}^{2}\ .
\end{equation}

Hence, the IM method has the following three-step procedures:

\emph{\textbf{A-step:}} Let \(F( \cdot )\) be the cdf of
\(Z/\sqrt{U}\). From (\ref{eq:3}), The association step of the IM produces an initial
candidate collection of \(\theta\),

\begin{equation}\label{eq:4}
\Theta_{X,W}(V) = \left\{ \theta:X = \theta + \sqrt{W}F^{- 1}(V) \right\},\ V\sim Unif(0,1).
\end{equation}

\emph{\textbf{P-step:}} We are interested in two-sided CIs, and the
construction methods for one-sided CIs are analogous. Hence, for a
singleton assertion \(A = \{\theta:\theta = \theta_{0}\}\), the valid
predictive random set (PRS) of \(V\) proposed by \cite{Martin2013}
is given by

\[
S = \left\{ V:|V - 0.5| \leq \left| \widetilde{V} - 0.5 \right| \right\},\ \widetilde{V}\sim Unif(0,1).
\]

\textbf{\emph{C-step}:} Combine \(\Theta_{X,W}(V)\) and \(S\) to obtain
a new candidate collection of \(\theta\), that is,

\[
\Theta_{X,W}(S) = \bigcup_{V \in S}^{}{\Theta_{X,W}(V)} = \left\lbrack X - \sqrt{W}F^{- 1}\left( 0.5 + \left| \widetilde{V} - 0.5 \right| \right),X - \sqrt{W}F^{- 1}\left( 0.5 - \left| \widetilde{V} - 0.5 \right| \right) \right\rbrack.
\]

Since \(\theta\) has a constraint
\(\theta \in C = \{\theta:\theta \geq 0\}\), conflict situation
\(X - \sqrt{W}F^{- 1}\left( 0.5 + \left| \widetilde{V} - 0.5 \right| \right) < 0\)
exists. To avoid the conflict case
\(\Theta_{X,W}(S)\bigcap_{}^{}C = \varnothing\), one possible way is to
enlarge the PRS \(\Theta_{X,W}(S)\) for \(\theta\) as follows:

\[
\Theta_{X,W}'(S) = \begin{cases}
\Theta_{X,W}(S) \cap C, & \text{if } \Theta_{X,W}(S) \cap C \neq \varnothing; \\
\left\{ 0 \right\},\ \ & \text{if } \Theta_{X,W}(S) \cap C = \varnothing.
\end{cases}
\]

Hence, the closed form of the final PRS \(\Theta_{X,W}'(S)\) is

\[
\Theta_{X,W}'(S) = \left\lbrack \max\left\{ 0,X - \sqrt{W}F^{- 1}\left( 0.5 + \left| \widetilde{V} - 0.5 \right| \right) \right\},\max\left\{ 0,X - \sqrt{W}F^{- 1}\left( 0.5 - \left| \widetilde{V} - 0.5 \right| \right) \right\} \right\rbrack.
\]

For an assertion \(A = \{\theta:\theta = \theta_{0}\}\), we calculate
the plausibility function when \(\theta_{0} = 0\),

\[
{pl}_{X,W}(A) = P_{S}{\left\{ \Theta_{X,W}'(S) \cap \{ 0\} \neq \varnothing \right\} = \begin{cases}
2\left( 1 - F\left( \frac{X}{\sqrt{W}} \right) \right), \\
1,\
\end{cases} }\begin{matrix}
\frac{1}{2} < F\left( \frac{X}{\sqrt{W}} \right) \leq 1; \\
otherwise,
\end{matrix}
\]

and for \(\theta_{0} > 0\),

\[
{pl}_{X,W}(A) = P_{S}\left\{ \Theta_{X,W}'(S) \cap A \neq \varnothing \right\} = \begin{cases}
2F\left( \frac{X - \theta_{0}}{\sqrt{W}} \right), \\
2\left( 1 - F\left( \frac{X - \theta_{0}}{\sqrt{W}} \right) \right),\
\end{cases} \begin{matrix}
F\left( \frac{X - \theta_{0}}{\sqrt{W}} \right) < \frac{1}{2}; \\
F\left( \frac{X - \theta_{0}}{\sqrt{W}} \right) \geq \frac{1}{2}.
\end{matrix}
\]

Given the significance level \(\alpha\), we confirm that the assertion
\(A = \{\theta:\theta = \theta_{0}\}\) is wrong if
\({pl}_{X,W}(A) \leq \alpha\). Moreover, the IM two-sided \(1 - \alpha\)
CI for \(\theta\) is

\[
\left\{ \theta:{pl}_{X,W}(A) > \alpha \right\} = (\max\left\{ 0,X - \sqrt{W}F^{- 1}\left( 1 - \frac{\alpha}{2} \right) \right\},\max\left\{ 0,X - \sqrt{W}F^{- 1}\left( \frac{\alpha}{2} \right) \right\}).
\]

\begin{theorem}
\label{thm:2}
The plausibility function \({pl}_{X,W}(A)\)
of the IM method is valid for assertion
\(A = \{\theta:\theta = \theta_{0}\}\) if, for each
\(\alpha \in (0,1)\),

\[
\sup_{\theta \in A}{P_{(X,W)|\theta,\sigma^{2}}\left\{ {pl}_{X,W}(A) \leq \alpha \right\}} \leq \alpha.
\]

Hence, the resulting IM CI can guarantee the nominal coverage
probability \(1 - \alpha\).
\end{theorem}

\subsection{Inference in the Poisson case}

Suppose that we observe \(X = B + S\) and \(W\), where \(B\), \(S\), and
\(W\) are independent random variables for which
\(B\sim Poisson(\varepsilon)\), \(S\sim Poisson(\lambda)\), and
\(W\sim Poisson(m\varepsilon)\), where \(\lambda,\varepsilon \geq 0\)
are unknown parameters and \(m > 0\) is known. Let
\(\theta = \varepsilon + \lambda\); the goal is to infer \(\lambda\),
where \(\lambda \in C = \{\lambda:\lambda \geq 0\}\). In this section,
we propose two new IM-based CIs for \(\lambda\).

\subsubsection{IM confidence interval}

Let \(F_{\theta}( \cdot )\) and \(F_{m\varepsilon}( \cdot )\) be the
cdfs of \(X\sim Poisson(\theta)\) and \(W\sim Poisson(m\varepsilon)\),
respectively. Owing to the discreteness of the Poisson distribution,
\cite{Leaf2012} gave an association linking \(X\),
\(\theta\), and an auxiliary variable \(U\sim P_{U}\) as follows:

\begin{equation}\label{eq:5}
F_{\theta}(X - 1) \leq U \leq F_{\theta}(X),\ U\sim Unif(0,1).
\end{equation}

Similarly, the association for \(W\), given \(m\varepsilon\), may be
written as

\begin{equation}\label{eq:6}
F_{m\varepsilon}(W - 1) \leq \widetilde{U} \leq F_{m\varepsilon}(W),\widetilde{U}\sim Unif(0,1),
\end{equation}

where \(U\) and \(\widetilde{U}\) are mutually independent. Let
\(G_{a,b}( \cdot )\) be the cdf of the Gamma distribution with shape
\emph{a} and scale \emph{b} and let \(G_{a,b}^{- 1}( \cdot )\) be the
inverse function of \(G_{a,b}( \cdot )\); then, from the well-known
relation between the Poisson and Gamma distributions,
\(G_{X,1}(\theta) = 1 - F_{\theta}(X - 1)\), we can rewrite the joint
associations (\ref{eq:5}) and (\ref{eq:6}) as follows:

\begin{equation}\label{eq:7}
G_{X,1}^{- 1}(1 - U) \leq \theta \leq G_{X + 1,1}^{- 1}(1 - U),U\sim Unif(0,1),
\end{equation}

\begin{equation}\label{eq:8}
G_{W,1}^{- 1}\left( 1 - \widetilde{U} \right) \leq m\varepsilon \leq G_{W + 1,1}^{- 1}\left( 1 - \widetilde{U} \right),\widetilde{U}\sim Unif(0,1),
\end{equation}

Hence, by (\ref{eq:5}) and (\ref{eq:6}), the initial association for \(\lambda = \theta - \varepsilon\) is

\begin{equation}\label{eq:9}
G_{X,1}^{- 1}(1 - U) - \frac{G_{W + 1,1}^{- 1}\left( 1 - \widetilde{U} \right)}{m} \leq \lambda{\leq G}_{X + 1,1}^{- 1}(1 - U) - \frac{G_{W,1}^{- 1}\left( 1 - \widetilde{U} \right)}{m}.
\end{equation}

The IM method has the following three-step procedures:

\emph{\textbf{A-step:}} Let \(K_{X,W + 1}( \cdot )\) and
\(K_{X + 1,W}( \cdot )\) be the distribution functions of the two
endpoints in (\ref{eq:9}). The association step of the IM produces an initial
candidate collection of \(\lambda\)

\begin{equation}\label{eq:10}
\Theta_{X,W}(V) = \left\{ \lambda:K_{X,W + 1}^{- 1}(V) \leq \lambda \leq K_{X + 1,W}^{- 1}(V) \right\},\ V\sim Unif(0,1).
\end{equation}

\emph{\textbf{P-step:}} For a singleton assertion
\(A = \{\lambda:\lambda = \lambda_{0}\}\), the valid PRS of \(V\)
proposed by \cite{Martin2013} is given by

\[
S = \left\{ V:|V - 0.5| \leq \left| \widetilde{V} - 0.5 \right| \right\},\ \widetilde{V}\sim Unif(0,1),\widetilde{V}\sim Unif(0,1).
\]

\textbf{\emph{C-step}:} \(\Theta_{X,W}(V)\) and \(S\) are combined to obtain
a random set \(\Theta_{X,W}(S) = \bigcup_{V \in S}^{}{\Theta_{X,W}(V)}\)
for \(\lambda\). Note that \(\lambda\) has a constraint
\(\lambda \in C = \{\lambda:\lambda \geq 0\}\). Similar to the
constrained normal case, to avoid the conflict case
\(\Theta_{X,W}(S)\bigcap_{}^{}C = \varnothing\), we choose to enlarge
the PRS \(\Theta_{X,W}(S)\) for \(\lambda\) as follows:

\[
\Theta_{X,W}'(S) = \begin{cases}
\Theta_{X,W}(S) \cap C, & if\ \Theta_{X,W}(S) \cap C \neq \varnothing; \\
\left\{ 0 \right\}, & if\ \Theta_{X,W}(S) \cap C = \varnothing.
\end{cases}
\]

Hence, the closed form of the final PRS \(\Theta_{X,W}'(S)\) is

\[
\Theta_{X,W}'(S) = \left\lbrack \max\left\{ 0,K_{X,W + 1}^{- 1}\left( 0.5 - \left| \widetilde{V} - 0.5 \right| \right) \right\},\max\left\{ 0,K_{X + 1,W}^{- 1}\left( 0.5 + \left| \widetilde{V} - 0.5 \right| \right) \right\} \right\rbrack.
\]

Moreover, we calculate the plausibility function for an assertion
\(A = \{\lambda:\lambda = \lambda_{0}\}\) as follows:

\[
{pl}_{X,W}(A) = P_{S}\left\{ \Theta_{X,W}'(S) \cap \left\{ 0 \right\} \neq \varnothing \right\} = \begin{cases}
2K_{X,W + 1}(0), & K_{X,W + 1}(0) < \frac{1}{2}; \\
1, & otherwise,
\end{cases}
\]

when \(\lambda_{0} = 0\), and

\[
{pl}_{X,W}(A) = P_{S}\left\{ \Theta_{X,W}'(S) \cap A \neq \varnothing \right\} = \begin{cases}
2K_{X,W + 1}\left( \lambda_{0} \right), & K_{X,W + 1}\left( \lambda_{0} \right) < \frac{1}{2}; \\
{2(1 - K}_{X + 1,W}(\lambda_{0})), & K_{X + 1,W}\left( \lambda_{0} \right) > \frac{1}{2}; \\
1, & otherwise,
\end{cases}
\]

when \(\lambda_{0} > 0\).

\begin{theorem}
\label{thm:3}
According to Theorem~\ref{thm:1}, the plausibility
function \({pl}_{X,W}(A)\) of the IM method is valid for assertion \(A\)
if, for each \(\alpha \in (0,1)\),

\[
\sup_{\lambda \in A}{P_{(X,W)|\lambda,\varepsilon}\left\{ {pl}_{X,W}(A) \leq \alpha \right\}} \leq \alpha.
\]

\end{theorem}

For any \(\alpha \in (0,1)\), if \({pl}_{X,W}(A) \leq \alpha\), then the
assertion \(A\) is wrong. Moreover, this plausibility function yields a
two-sided IM \(1 - \alpha\) CI

\[
{CI}_{IM} = \left\{ \lambda:{pl}_{X,W}(A) > \alpha \right\} = (\max\left\{ 0,K_{X,W + 1}^{- 1}\left( \frac{\alpha}{2} \right) \right\},\max\left\{ 0,K_{X + 1,W}^{- 1}\left( 1 - \frac{\alpha}{2} \right) \right\}).
\]

Notably, deriving the closed-form expressions of
\(K_{X,W + 1}( \cdot )\) and \(K_{X + 1,W}( \cdot )\) is
challenging; thus, we recommend approximating them via a Monte Carlo
method, as detailed in Algorithm ~\ref{alg:IM_CI} below:
 
\begin{algorithm}[H]
    \setstretch{1.2} 
    \caption{Calculating the IM confidence interval $CI_{IM} = [\lambda_{L}, \lambda_{U}]$}
    \label{alg:IM_CI}
    
    \KwIn{Confidence level $1 - \alpha$, scale parameter $m > 0$, data $X$ and $W$, number of simulations $N$}
    \KwOut{IM interval $[\lambda_{L}, \lambda_{U}]$}
    
    \BlankLine
    \text{Randomly sample $U$ and $\widetilde{U}$ independently $N$ times from $Unif(0,1)$}\;
    
    \For{$i \leftarrow 1$ \KwTo $N$}{
        Calculate: \\
        $\lambda_{1,i} = \max\left\{ 0, G_{X,1}^{-1}(1 - U_i) - \frac{G_{W+1,1}^{-1}(1 - \widetilde{U}_i)}{m} \right\}$\;
        $\lambda_{2,i} = \max\left\{ 0, G_{X+1,1}^{-1}(1 - U_i) - \frac{G_{W,1}^{-1}(1 - \widetilde{U}_i)}{m} \right\}$\;
    }
    
    Approximate $\lambda_{L}$ as the $\alpha/2$ quantile of the $N$ realizations $\{\lambda_{1,i}\}_{i=1}^N$\;
    Approximate $\lambda_{U}$ as the $1 - \alpha/2$ quantile of the $N$ realizations $\{\lambda_{2,i}\}_{i=1}^N$\;
    
    \Return $[\lambda_{L}, \lambda_{U}]$\;
\end{algorithm}

\subsubsection{Nonrandomized IM confidence interval using a random weighting approach}

In general, the association model (\ref{eq:10}) of the IM is an interval, which
may result in a conservative CI. Certain adjustments are required to
address this discreteness. Inspired by the random weighting idea \citep{Lu2019}, we consider a nonrandomized IM (NIM) approach to modify
inequalities (\ref{eq:5}-\ref{eq:6}) to accurate equations so that we can improve the
accuracy of the candidate value of \(\lambda\). Specifically, let
\(F_{\theta}( \cdot )\) be the cdf of the Poisson distribution with
parameter \(\theta\); for any fixed \(\omega \in \lbrack 0,1\rbrack\)
and a nonnegative integer \(X = x \in \lbrack 0, + \infty)\), let

\[
J_{x,\omega}(\theta) = \begin{cases}
\omega I_{0}(\theta) + (1 - \omega)F_{\theta}(x), & if\ x = 0, \\
\omega F_{\theta}(x - 1) + (1 - \omega)F_{\theta}(x), & if\ x \in (0,\infty),
\end{cases}
\]

where \(I_{D}( \cdot )\) is the indicator function of set \(D\). It is
a strictly decreasing function of \(\theta \in \lbrack 0,\infty)\) and
has a range of {[}0,1{]}. Moreover, given \(x\), for every \(\omega\),
\(u \in Unif(0,1)\), there exists a unique solution

\begin{equation}\label{eq:11}
\theta = J_{x,\omega}^{- 1}(u) = sup\left\{ \theta:J_{x,\omega}(\theta) \geq u \right\}.
\end{equation}

Similarly, for any fixed \(m > 0\),
\(\widetilde{\omega} \in \lbrack 0,1\rbrack\) and a nonnegative integer
\(W = w \in \lbrack 0, + \infty)\), let

\begin{equation}\label{eq:12}
{\widetilde{J}}_{w,\widetilde{\omega}}(\varepsilon) = \begin{cases}
\widetilde{\omega}I_{0}(m\varepsilon) + \left( 1 - \widetilde{\omega} \right)F_{m\varepsilon}(w), & if\ w = 0, \\
\widetilde{\omega}F_{m\varepsilon}(w - 1) + \left( 1 - \widetilde{\omega} \right)F_{m\varepsilon}(w), & if\ w \in (0,\infty).
\end{cases}
\end{equation}

For every \(\widetilde{\omega}\), \(\widetilde{u} \in Unif(0,1)\), there
exists a unique solution

\[
\varepsilon = {\widetilde{J}}_{w,\widetilde{\omega}}^{- 1}\left( \widetilde{u} \right) = sup\{\varepsilon:{\widetilde{J}}_{w,\widetilde{\omega}}(\varepsilon) \geq \widetilde{u}\}.
\]

From (\ref{eq:11}) and (\ref{eq:12}), we have
\(\lambda = \theta - \varepsilon = J_{x,\omega}^{- 1}(u) - {\widetilde{J}}_{w,\widetilde{\omega}}^{- 1}\left( \widetilde{u} \right)\).
Let \(H_{x,w}( \cdot )\) be its distribution function, the NIM CI
has the following three steps:

\emph{\textbf{A'-step:}} The new NIM method produces a new association
for \(\lambda\), that is,

\[
\Theta_{x,w}(v) = \left\{ \lambda:\lambda = H_{x,w}^{- 1}(v) \right\},\ v\sim Unif(0,1).
\]

\emph{\textbf{P'-step:}} For a singleton assertion
\(A = \{\lambda:\lambda = \lambda_{0}\}\), the valid PRS for \(v\) is

\[
S = \left\{ v:|v - 0.5| \leq |V - 0.5| \right\},V\sim Unif(0,1).
\]

\textbf{\emph{C'-step}:} To obtain an initial PRS for \(\lambda\),
\(\Theta_{x,w}(v)\) and \(S\) are combined as

\[
\Theta_{x,w}(S) = \bigcup_{v \in S}^{}{\Theta_{x,w}(v)} = \left\lbrack H_{x,w}^{- 1}\left( 0.5 - |V - 0.5| \right),H_{x,w}^{- 1}\left( 0.5 + |V - 0.5| \right) \right\rbrack.
\]

Since \(\lambda \in C = \{\lambda:\lambda \geq 0\}\), the final PRS for
\(\lambda\) is obtained by enlarging \(\Theta_{x,w}(S)\) to

\[
\Theta_{x,w}'(S) = \left\lbrack {max\{ 0,H}_{x,w}^{- 1}\left( 0.5 - |V - 0.5| \right)\},max\{ 0,H_{x,w}^{- 1}\left( 0.5 + |V - 0.5| \right)\} \right\rbrack.
\]

Then, the plausibility function of an assertion
\(A = \{\lambda:\lambda = \lambda_{0}\}\), is given by

\[
{pl}_{x,w}(A) = P_{S}\left\{ \Theta_{x,w}'(S) \cap \{ 0\} \neq \varnothing \right\} = \begin{cases}
2H_{x,w}(0), & H_{x,w}(0) < \frac{1}{2}; \\
1,\  & otherwise,
\end{cases}
\]

when \(\lambda_{0} = 0\), and

\[
{pl}_{x,w}(A) = P_{S}\left\{ \Theta_{x,w}'(S) \cap A \neq \varnothing \right\} = \begin{cases}
2H_{x,w}\left( \lambda_{0} \right), & H_{x,w}\left( \lambda_{0} \right) < \frac{1}{2}; \\
2\left( 1 - H_{x,w}\left( \lambda_{0} \right) \right), & H_{x,w}\left( \lambda_{0} \right) \geq \frac{1}{2},
\end{cases}
\]

when \(\lambda_{0} > 0\).

\begin{theorem}
\label{thm:4}
Let \(S\sim P_{S}\) be a valid predictive
random set for \(v\sim Unif(0,1)\), that is,
\(P_{S}\left\{ v \in S \right\} \geq_{st}Unif(0,1)\), where
\(" \geq_{st}"\) means ``stochastically no smaller than''. If
\(H_{x,w}(\lambda)\sim Unif(0,1)\) for \((x,w)\sim P_{(x,w)|\lambda}\)
for all \(\lambda\), then the NIM method is valid.
\end{theorem}
Given the significance level \(\alpha\), the resulting two-sided
\(1 - \alpha\) NIM CI for \(\lambda\) is

\[
{CI}_{NIM} = \left\{ \lambda:{pl}_{x,w}(A) > \alpha \right\} = \left( \max\left\{ 0,H_{x,w}^{- 1}\left( \frac{\alpha}{2} \right) \right\},\max\left\{ 0,H_{x,w}^{- 1}\left( 1 - \frac{\alpha}{2} \right) \right\} \right).
\]

According to Theorem~\ref{thm:4}, the NIM method is valid if the Monte Carlo
approximation \(H_{x,w}(\lambda)\) exactly follows a uniform
distribution on (0,1). Although this condition is not always satisfied,
the simulation study in Section\,4.2 shows that \(H_{x,w}(\lambda)\)
possesses certain desirable properties, which can provide a theoretical
justification for the superior coverage performance of the NIM
intervals. Moreover, since it is difficult to obtain the closed form of
\(H_{x,w}( \cdot )\), we suggest approximating it via a Monte Carlo
method as follows:

\begin{algorithm}[H]
    \setstretch{1.2} 
    \caption{Calculating the NIM confidence interval $CI_{NIM} = [\lambda_{L}, \lambda_{U}]$}
    \label{alg:NIM_CI}
    
    \KwIn{Confidence level $1-\alpha$, $m > 0$, data $X=x$, $W=w$, number of simulations $N$}
    \KwOut{NIM interval $[\lambda_{L}, \lambda_{U}]$}
    
    \BlankLine
    \For{$i \leftarrow 1$ \KwTo $N$}{
        1. Sample $\omega, \widetilde{\omega}, u, \widetilde{u} \overset{iid}{\sim} Unif(0,1)$\;
        2. Solve for $\theta$ using: \\
        $u = \begin{cases} 
        \omega I_{0}(\theta) + (1 - \omega)F_{\theta}(x), & \text{if } x = 0 \\
        \omega F_{\theta}(x - 1) + (1 - \omega)F_{\theta}(x), & \text{if } x \in (0,\infty)
        \end{cases}$\;
        
        3. Solve for $\varepsilon$ using: \\
        $\widetilde{u} = \begin{cases} 
        \widetilde{\omega}I_{0}(m\varepsilon) + (1 - \widetilde{\omega})F_{m\varepsilon}(w), & \text{if } w = 0 \\
        \widetilde{\omega}F_{m\varepsilon}(w - 1) + (1 - \widetilde{\omega})F_{m\varepsilon}(w), & \text{if } w \in (0,\infty)
        \end{cases}$\;
        
        4. Calculate $\lambda^{(i)} = \max\{0, \theta - \varepsilon\}$\;
    }
    
    Calculate $\lambda_{L}$ as the $\alpha/2$ quantile of the realizations $\{\lambda^{(i)}\}_{i=1}^N$\;
    Calculate $\lambda_{U}$ as the $1 - \alpha/2$ quantile of the realizations $\{\lambda^{(i)}\}_{i=1}^N$\;
    
    \Return $[\lambda_{L}, \lambda_{U}]$\;
\end{algorithm}

\begin{remark}
The solutions of \(\theta\) and
\(\varepsilon\) in Step 1 can be achieved via either the
\textquotesingle uniroot\textquotesingle{} function in R software or via
the bisection method function in Python software. The accuracy of the
approximate NIM solution depends only on the number of repetitions
\emph{N}. We recommend \emph{N} =10,000 in practical applications and
use this value in the following simulation studies to ensure that there
is a greater than 95\% probability of an absolute error of less than
0.01.
\end{remark}
\begin{remark}
When evaluating the coverage performance of
the NIM CIs, if the number of simulated experimental data runs is
relatively large (e.g., 10,000 runs), it is necessary to obtain
solutions to 10,000 nonlinear equations separately for each batch of
generated data. This is extremely challenging for the uniroot function
in R software. This is because the uniroot function adopts a serial
computing strategy, whereby it must solve one equation before proceeding
to the next, which typically entails substantial additional
computational time. To address this issue, we adopt the strategy of
solving equations via the bisection method in Python while leveraging
graphics processing units (GPUs) for large-scale parallel computing.
Specifically, we only need to invoke the PyTorch package and utilize the
concept of tensors to effortlessly achieve one-time processing of
large-scale data generation and equation solving tasks. This
high-efficiency computing approach can greatly reduce the time cost
required by the uniroot function. For example, when the empirical
coverage rate of the NIM interval at the 90\% confidence level is
calculated, with 10,000 simulated experimental data runs assumed and the
parameters set as \(m = 20\), \(\lambda = 1\), and \(\varepsilon = 3\),
the coverage rates calculated via R and Python are 0.9106 and 0.9083,
respectively, with corresponding time costs of approximately 11 hours
and 15 seconds. Clearly, the latter not only ensures accurate
computation but also greatly decreases the computational time; thus, we
recommend adopting a GPU-based parallel computing strategy to address
the problem of large-scale simulation computing. The R code and Python code can be found in the supplementary materials.
\end{remark}
\section{Simulation study}

\subsection{Constrained normal case}

For the constrained normal mean \(\theta\), we have proven that the IM
CI has the exact coverage probability. There is no need to perform
simulations to demonstrate its validity. However, the complex
expressions of Bayesian credible intervals make it difficult to
theoretically compare the frequentist coverage performance between
Bayesian and IM intervals. In this section, we perform Monte Carlo
simulations to examine the frequentist performance of the two CIs. We
mainly report the empirical coverage probability and expected length.
These two measurements are commonly used to evaluate the reliability and
precision of an interval. Regarding the setting of parameter values, we
refer to \cite{Zhang2003} and \cite{Leaf2012}.
Setting \(\sigma^{2} = 1\), \(\theta = 0.0\ (0.1)\ 4.0\),
\(1 - \alpha = 0.90,\ 0.95\), and \(r = 5,\ 10,\ 20,\ 50.\) For every
given \((\theta,\sigma^{2},r,\alpha)\), we first independently resample
\(X\sim N(\theta,\sigma^{2})\) and \(W\sim\sigma^{2}\chi_{r}^{2}\)
10,000 times, calculate the two CIs of \(\theta\) accordingly, and
compute the corresponding frequencies that cover \(\theta\). We regard
the coverage frequency as the coverage probability. For the expected
length, by definition,

\[
Expected\ length = E_{\theta,\sigma^{2}}(length(CI(x,w)))
\]

\[
\ \ \ \ \ \ \ \ \ \ \ \ \ \ \ \ \ \ \ \ \ \ \ \ \ \ \ \ \ \ \ \ \ \ \ \ \ \ \ \ \ \ \ \ \ \ \ \ \ \ \  = \frac{\sum_{i = 1}^{M}\left( U\left( x,w;\theta,\sigma^{2} \right) - L\left( x,w;\theta,\sigma^{2} \right) \right)}{M},
\]

where \(M = 10,000\), \(U\left( x,w;\theta,\sigma^{2} \right)\) and
\(L\left( x,w;\theta,\sigma^{2} \right)\) are the upper and lower limits
of the interval for \(x\) and \(w\), respectively. The simulation
results are shown in Figures~\ref{fig:1} to~\ref{fig:2}.

\textbf{(Figures~\ref{fig:1} to~\ref{fig:2} are about here)}

From Figures~\ref{fig:1} to~\ref{fig:2}, the IM intervals consistently
maintain stable coverage probabilities that closely align with the
nominal levels (0.90 and 0.95) for all values of \(\theta\), indicating
that the IM intervals exhibit exact coverage performance. In contrast,
the Bayesian intervals exhibited a nonmonotonic pattern: the coverage
probability of the Bayesian intervals first increased to a peak, then
decreased sharply to a level significantly lower than that of the IM
intervals, and subsequently gradually recovered and converged toward the
nominal confidence level as the parameter \(\theta\) increased.
Furthermore, as the degrees of freedom \(r\) increased, this
undercoverage phenomenon of Bayesian intervals was barely mitigated. In
terms of expected lengths, the expected lengths of both types of
intervals monotonically increased with \(\theta\), and the expected
length of the IM intervals was consistently longer than that of the
Bayesian intervals. Notably, the coverage probability of the Bayesian
interval is significantly lower than the nominal confidence level in
many simulation scenarios, indicating that the IM method accurately
controls the coverage probability by appropriately increasing the
interval length. Moreover, the gap in expected lengths between the two
types of intervals gradually narrowed as \(r\) increased. For example,
when \(r = 50\), the expected lengths of the two intervals were almost
identical (Figure~\ref{fig:2}). Therefore, IM intervals achieve more
reliable frequentist coverage at the cost of slightly longer interval
lengths. Hence, the proposed IM CIs outperform the Bayesian intervals.

\subsection{Constrained Poisson case}

The lack of explicit expressions for Bayesian, IM, and NIM CIs precludes
direct theoretical comparisons of their coverage performance.
Therefore, Monte Carlo simulations are performed herein to evaluate the
frequentist performance of these intervals, with coverage probability
and expected length as the evaluation criteria. The parameter settings
adopted from \cite{Zhang2003} and \cite{Leaf2012}
are as follows: \(\varepsilon = 3.0\); \(\lambda = 0.0\ (0.1)\ 10.0\);
\(1 - \alpha = 0.90,\ 0.95\); and \(m = 20,\ 50,\ 100,\ 300.\) The
simulation results are shown in Figures~\ref{fig:3} to~\ref{fig:4}.

\textbf{(Figures~\ref{fig:3} to ~\ref{fig:4} are about here)}

From Figures~\ref{fig:3} to~\ref{fig:4}, as \(\lambda\) and \(m\)
increase, the coverage probability of the Bayesian intervals becomes
highly unstable, exhibiting an overall monotonically decreasing trend.
This leads to a significant decrease below the preset confidence level
in certain scenarios. In contrast, the proposed IM intervals
consistently maintain stable coverage probabilities that are above the
nominal confidence levels (0.90 and 0.95) for all values of \(\lambda\)
and \(m\), validating the effectiveness guaranteed by
Theorem~\ref{thm:3}. Moreover, as an improved variant of the
conservative IM interval, the NIM intervals exhibit smaller
fluctuations in coverage probability and align most closely with the
nominal confidence levels, demonstrating superior stability compared
with both the Bayesian and IM intervals---this can be attributed to the
correction of the discreteness of the Poisson distribution via the
random weighting method.

In terms of the expected length, all three intervals exhibit a
monotonically increasing trend with increasing \(\lambda\) and the
expected length of the IM intervals is consistently slightly longer than
that of the Bayesian intervals, which constitutes a reasonable trade-off
for ensuring stable coverage. Note that as \(m\) increases, the
differences in the expected length among the three intervals show almost
no significant change, indicating that these intervals are robust in
terms of interval width. Moreover, for all values of \(m\), the
expected lengths of the NIM and Bayesian intervals are remarkably close;
for smaller values of \(\lambda\ ( \leq 1.0)\), the expected length of
the NIM is shorter, whereas conversely, that of the Bayesian interval
is slightly shorter. In summary, the IM interval provides stable and
reliable interval estimation for constrained Poisson parameters, as it
does not rely on prior information and can address the issue of overly
conservative behavior and insufficient coverage in Bayesian intervals.
The coverage performance of the NIM method is superior to that of
Bayesian intervals. Building upon the IM interval, the NIM interval
optimizes the interval length to be close to or even shorter than that
of the Bayesian interval, resulting in coverage probabilities that are
closest to the nominal confidence levels. Therefore, both the IM and NIM
intervals are more suitable for practical application scenarios than the
Bayesian interval is.

\textbf{(Figures~\ref{fig:5} to~\ref{fig:8} are about here)}

According to Theorem~\ref{thm:4}, for each \(\alpha \in (0,1)\), the NIM method is
valid if the Monte Carlo approximation
\(H_{x,w}(\lambda)\) follows the uniform
distribution in (0,1). However, our simulation studies show that this
condition does not always hold, but the approximation
\(H_{x,w}(\lambda)\) can also provide some good
properties. Without loss of generality, we construct different pairs
\((\varepsilon,\ \lambda,m)\) to characterize the distribution function
of \(H_{x,w}(\lambda)\), where \(\varepsilon = 3.0\);
\(\lambda = 0.0,\ 0.5,\ 1.0,\ 2.0,\ 3.0,\ 4.0\); and
\(m = 20,\ 50,\ 100,\ 300.\) Specifically, for each \(\lambda\) and
\(\varepsilon\), we generate 10,000 Poisson random samples \((x,w)\) and
obtain a Monte Carlo estimate of the distribution function of
\(H_{x,w}(\lambda)\). Clearly, for small to moderate values of \(m\),
Figures~\ref{fig:5} and~\ref{fig:6} demonstrate that the approximate
\(H_{x,w}(\lambda)\)
is sufficiently close to Unif(0,1). Moreover, for large \(m\), Figures
~\ref{fig:7} and~\ref{fig:8} show that the approximate
\(H_{x,w}(\lambda)\) has the
following desirable property:

\[
\sup_{\lambda \in A}{P_{(x,w)|\lambda,\varepsilon}\left\{ {pl}_{x,w}(A) \leq \alpha \right\} = \sup_{\lambda \in A}{P_{(x,w)|\lambda,\varepsilon}\left\{ H_{x,w}(\lambda) \leq \frac{\alpha}{2}\bigcup_{}^{}{H_{x,w}(\lambda) \geq 1 - \frac{\alpha}{2}} \right\}}}
\]

\[
\leq \sup_{\lambda \in A}{P_{(x,w)|\lambda,\varepsilon}\left\{ V \leq \frac{\alpha}{2}\bigcup_{}^{}{V \geq 1 - \frac{\alpha}{2}} \right\} = \alpha,}
\]

where \(V\sim Unif(0,1)\). Hence, the approximate plausibility function
of the NIM is valid.

\section{Real data analysis}

In this section, we illustrate the application of the proposed intervals
with two real examples.

\emph{Example 1}. Neutrino mass inference

The neutrino mass is one of the core unresolved issues in particle
physics. For many years, researchers have universally believed that
neutrinos are massless. However, a Super-Kamiokande experiment in Japan
\citep{Fukuda1998} first confirmed that neutrinos have nonzero masses
through neutrino oscillation. In recent years, Germany's KATRIN
experiment \citep{Nucciotti2022} improved the precision of the neutrino
mass to 0.8 electron volts (eV) and further reduced it to 0.45 eV
\citep{Aker2025}. Using the constrained normal mean model, we set
\(X = 0.45\). Since existing studies have not provided measured values
of \(W\), we consider several possible values of \(W\), such as
\(W = 0.01,\ 0.1,\ 0.5,\ 1.0,\ 5.0,\ 10\), which cover small, medium,
and large value types. For the specification of \(r\), we follow
\cite{Zhang2003} and consider four cases: \(r = 5\), 10, 20, and 50.
We employ the Bayesian and IM methods to compute the 90\% and 95\% CIs
for the neutrino mass, with the corresponding results summarized in
Table~\ref{tab:1}.

\textbf{(Table~\ref{tab:1} is about here)}

Table~\ref{tab:1} presents the 0.90 and 0.95 CIs and their widths for neutrino
mass derived from the Bayesian and IM methods under different
combinations of observations \(W\) and degrees of freedom \(r\).
Clearly, the lengths of both the Bayesian and the IM CIs decrease with
increasing \(r\), and the length of the IM interval is shorter than that
of the Bayesian interval in most cases. Notably, the length of the
Bayesian interval is significantly shorter than that of the IM interval
in certain cases; for example, when \(W = 10\) and \(r = 5\), the level
0.90 Bayesian and IM CIs are {[}0.0000, 3.0651{]} and {[}0.0000,
6.2209{]}, respectively, and the corresponding level 0.95 CIs are
{[}0.0000, 3.8400{]} and {[}0.0000, 8.9531{]}. This phenomenon may lead
to the coverage probability of the Bayesian interval failing to meet the
prespecified confidence level, as shown in Figures~\ref{fig:1}
and~\ref{fig:2}. In contrast, the IM interval can satisfy exact
coverage probability requirement. More importantly, unlike the Bayesian
interval, the plausibility function of the IM method can provide more
information for unknown parameters \(\theta\). For instance, given the significance
level \(\alpha = 0.1\), when \(X = 0.45\), \(W = 1.0\) and \(r = 10\),
Figure~\ref{fig:9}–(a) demonstrates that each point \(\theta\) in the IM interval
{[}0.2324, 0.6659{]} is individually sufficiently plausible. Therefore,
we recommend the use of the IM CI in the detection of neutrino masses.

\textbf{(Figure~\ref{fig:9} is about here)}

\emph{Example 2}. Estimation of neutrino signal strength

As noted by \citet{Dawn2017}, approximately 100 trillion neutrinos pass
through the human body each second. However, these ghostly subatomic
particles cannot be detected by our senses, as they interact only via
weak forces, which makes their detection exceedingly challenging. To
date, detecting neutrino signals remains a significant research topic in
high-energy physics. Nevertheless, in the presence of background
interference, experiments may record observations with very low
probability, such as \(X = 0\). In such cases, classical Neyman CIs may
yield empty sets, making it particularly important to construct accurate
CIs for low-probability observations. On the basis of \cite{Eitel1999},
the expected background \(\varepsilon\) is \(2.88 \pm 0.13\); here, we
consider low-probability observations \(X = 0,\ 1\). For the observation
\(W\) and the parameter \(m\), according to \cite{Zhang2003},
we consider \(W = 10,\ 20,\ 30,\ 40\) and \(m = 20,\ 50,\ 100,\ 300\).
Using the constrained Poisson mean model, we compute the Bayesian, IM,
and NIM CIs for the neutrino signal strength at the 90\% and 95\%
confidence levels, respectively, with the corresponding results
summarized in Table~\ref{tab:2}.

\textbf{(Table~\ref{tab:2} is about here)}

As shown in Table~\ref{tab:2}, when \(X = 0\), the Bayesian CI remains unchanged
across different parameter settings, indicating its insensitivity to
observed data. In contrast, both the IM and NIM intervals provide
flexible CIs for the parameter of interest under different combinations
of observed values, with the NIM interval consistently exhibiting the
shortest interval length. Furthermore, when \(X = 1\), the
Bayesian interval has the smallest width in most cases compared with the
IM and NIM intervals. Although the IM interval can be narrower than the
Bayesian method in certain scenarios---for example, when \(X = 1\),
\(m = 20\), and \(W = 40\), the 90\% Bayesian and IM confidence
intervals are {[}0.00, 3.02{]} and {[}0.00, 2.78{]}, respectively---the
width of the IM interval is generally significantly greater than that of
the Bayesian interval. This is an inherent consequence of the
conservative inequality-constrained model adopted by the IM method to
guarantee theoretical coverage, whereas the Bayesian interval may suffer
from undercoverage, a conclusion also supported by Figures~\ref{fig:3}
to~\ref{fig:4}.

Notably, as an improved version of the conservative IM interval, the NIM
method consistently yields shorter widths than the IM interval for
\(X = 1\). Its width is closer to that of the Bayesian interval and can
even be narrower in some cases. For instance, when \(X = 1\),
\(m = 20\), and \(W = 40\), the 90\% Bayesian and NIM confidence
intervals are {[}0.00, 2.78{]} and {[}0.00, 2.14{]}, respectively.
Combined with the simulation analysis in Section 4.2, the NIM interval
not only addresses the issues of unstable coverage performance and
insensitivity to extreme data in the Bayesian interval but also
mitigates the length redundancy observed in the IM interval, thereby
enhancing its practical utility for estimating neutrino signal strength
in low-probability observation scenarios.

To better understand the difference between the IM and NIM CIs,
Figure~\ref{fig:9}–(b) shows their plausibility functions. By setting a significance
level \(\alpha \in (0,1)\), the corresponding CIs of the two methods can
be straightforwardly obtained from the plausibility function. For
example, when \(\alpha = 0.05\), \(X = 0\), \(W = 10\), and \(m = 20\),
the IM and NIM CIs are {[}0.00, 3.19{]} and {[}0.00, 2.49{]},
respectively. Clearly, the NIM interval is shorter than the IM interval.
More importantly, these CIs derived under the inferential model
framework have a desirable property: every point inside the interval has
an intuitive plausibility measure. In contrast, the Bayesian CI cannot
provide such a measure and is generally longer than the NIM interval.
Therefore, we recommend the use of the NIM method for estimating
neutrino signal strength in practical applications.

\section{Discussion}

Building upon the IM framework, this study proposes interval estimation
methods for constrained parameters in normal and Poisson distributions
with unknown nuisance parameters. The core contribution lies in
addressing key shortcomings of conventional approaches. Compared with
Bayesian methods, which exhibit undercoverage in the normal and Poisson
cases, the IM intervals integrate parameter constraints with predictive
random sets to ensure stable coverage closely aligned with nominal
levels. The NIM intervals further correct for the discreteness of
Poisson distributions via random weighting, improving inferential
efficiency while achieving superior interval lengths, particularly under
low-count scenarios where they significantly outperform Bayesian
intervals. Additionally, the proposed IM-based methods are prior-free,
and each parameter value within the interval is accompanied by a
likelihood-based measure, providing more physically interpretable
inference for fields such as high-energy physics.

From a practical perspective, the IM and NIM methods demonstrate
considerable value in cutting-edge scientific problems such as neutrino
mass measurement and signal strength estimation. They not only resolve
issues where traditional intervals may be empty or overly conservative
under low-count observations, but also offer refined characterization of
parameter uncertainty through likelihood functions, thereby delivering
more reliable statistical support for experimental design and result
interpretation. As a prior-free probabilistic inferential framework,
this approach not only enriches the methodological system for
constrained parameter inference but also offers an effective new pathway
for analyzing complex data in fields such as high-energy physics and
astronomical observation.

Despite these advances, limitations remain in the inference for
constrained Poisson models with nuisance parameters. Specifically, the
exact distribution function of the parameter of interest in the joint
linking model of data, parameters, and auxiliary variables---based on
the random weighting idea---is approximate, and better methods are
required to further improve frequentist coverage control. Future
research could focus on extending the proposed methods to other
constrained parameter scenarios, such as exponential and multinomial
distributions, as well as investigating inference problems involving
multiple parameter constraints or data-dependent settings.

\textbf{Data availability statement:} The data are contained within the
article.

\textbf{Conflict of interest:} The authors declare no conflicts of
interest.

\newpage
\appendix

\section*{Appendix A: Technical details}

\textbf{Proof of Theorem 1}

To prove Theorem 1, it suffices to show that the IM is valid, which has already been established in Theorem 2 of \cite{Martin2013}.

\textbf{Proof of Theorem 2}

From (\ref{eq:4}), \(F\left( \frac{X - \theta}{\sqrt{W}} \right)\sim Unif(0,1)\); then,

\[\sup_{\theta \in A}{P_{(X,W)|\theta,\sigma^{2}}\left\{ {pl}_{X,W}(A) \leq \alpha|\theta_{0} = 0 \right\}}\]

\[= \sup_{\theta \in A}{P_{(X,W)|\theta,\sigma^{2}}\left\{ F\left( \frac{X}{\sqrt{W}} \right) \geq 1 - \frac{\alpha}{2},\frac{1}{2} < F\left( \frac{X}{\sqrt{W}} \right) \leq 1|\theta_{0} = 0 \right\}} = \frac{\alpha}{2},\]

when \(\theta_{0} = 0\), and

\[\sup_{\theta \in A}{P_{(X,W)|\theta,\sigma^{2}}\left\{ {pl}_{X,W}(A) \leq \alpha|\theta = \theta_{0} \right\}}\]

\[
\begin{aligned}
= \sup_{\theta \in A}\Bigg\{ 
& P_{(X,W)|\theta,\sigma^{2}}\left\{ 
F\left( \frac{X - \theta_{0}}{\sqrt{W}} \right) \leq \frac{\alpha}{2},
F\left( \frac{X - \theta_{0}}{\sqrt{W}} \right) < \frac{1}{2}
\mid \theta = \theta_{0} 
\right\} \\
&+ P_{(X,W)|\theta,\sigma^{2}}\left\{ 
F\left( \frac{X - \theta_{0}}{\sqrt{W}} \right) \geq 1 - \frac{\alpha}{2},
F\left( \frac{X - \theta_{0}}{\sqrt{W}} \right) \geq \frac{1}{2}
\mid \theta = \theta_{0} 
\right\}
\Bigg\}
\end{aligned}
\]

\[= \alpha,\]

when \(\theta_{0} > 0\). Then, \(\sup_{\theta \in A}{P_{(X,W)|\theta,\sigma^{2}}\left\{ {pl}_{X,W}(A) \leq \alpha \right\}} \leq \alpha.\) By Theorem 1, the proof is complete.

\textbf{Proof of Theorem 3}

From (\ref{eq:10}), we have \(K_{X,W + 1}(\lambda) \leq V \leq K_{X + 1,W}(\lambda),\ \ V\sim Unif(0,1)\). Given any \(\alpha \in (0,1)\),

\[\sup_{\lambda \in A}P_{(X,W)|\lambda,\varepsilon}\left\{ {pl}_{X,W}(A) \leq \alpha{|\lambda}_{0} = 0 \right\}\]

\[{= \sup_{\lambda \in A}P}_{(X,W)|\lambda,\varepsilon}\left\{ 2K_{X,W + 1}(0) \leq \alpha,K_{X,W + 1}(0) < \frac{1}{2}|\lambda_{0} = 0 \right\} = \frac{\alpha}{2},\]

when \(\lambda_{0} = 0\), and

\[P_{(X,W)|\lambda,\varepsilon}\left\{ {pl}_{X,W}(A) \leq \alpha{|\lambda = \lambda}_{0} \right\}\]

\[= P_{(X,W)|\lambda,\varepsilon}\left\{ 2K_{X,W + 1}\left( \lambda_{0} \right) \leq \alpha,K_{X,W + 1}\left( \lambda_{0} \right) < \frac{1}{2} \right\} + P_{(X,W)|\lambda,\varepsilon}\left\{ 2 - 2K_{X + 1,W}\left( \lambda_{0} \right) \leq \alpha,K_{X + 1,W}\left( \lambda_{0} \right) > \frac{1}{2} \right\}\]

\[{\leq P_{(X,W)|\lambda,\varepsilon}\left\{ V \leq \frac{\alpha}{2},V < \frac{1}{2} \right\} + P}_{(X,W)|\lambda,\varepsilon}\left\{ V \geq 1 - \frac{\alpha}{2},V > \frac{1}{2} \right\},\]

when \(\lambda_{0} > 0\). Afterward, \(\sup_{\lambda \in A}{P_{(X,W)|\lambda,\varepsilon}\left\{ {pl}_{X,W}(A) \leq \alpha \right\}} \leq \alpha.\) Hence, the proof is complete.

\textbf{Proof of Theorem 4}

Given any \(\alpha \in (0,1)\), since \(H_{x,w}(\lambda)\sim Unif(0,1)\) for \((x,w)\sim P_{(x,w)|\lambda,\varepsilon}\) for all \(\lambda\), and

\[{pl}_{x,w}(A) = P_{S}{\left\{ \Theta_{x,w}(S) \cap A \neq \varnothing \right\} = P_{S}{\left\{ H_{x,w}(\lambda) \in S \right\} = P_{S}{\left\{ v \in S \right\}.}}}\]

Moreover, the predictive random set \(S\sim P_{S}\) is valid; that is, \(P_{S}\left\{ v \in S \right\} \geq_{st}Unif(0,1)\).

Hence, \(\sup_{\lambda \in A}{P_{(x,w)|\lambda,\varepsilon}\left\{ {pl}_{x,w}(A) \leq \alpha \right\} = \sup_{\lambda \in A}{P_{(x,w)|\lambda,\varepsilon}\left\{ P_{S}\left\{ v \in S \right\} \leq \alpha \right\}}}\)

\[\ \ \ \ \ \ \ \ \ \ \  \leq \sup_{\lambda \in A}{P_{(x,w)|\lambda,\varepsilon}\left\{ Unif(0,1) \leq \alpha \right\}} = \alpha.\]

Hence, the proof is complete.

\newpage
\section*{Appendix B: Figures 1 to 9 and Tables 1 to 2}

\clearpage
\refstepcounter{figure}
\begin{center}
\includegraphics[width=1.0\textwidth,height=0.95\textheight,keepaspectratio]{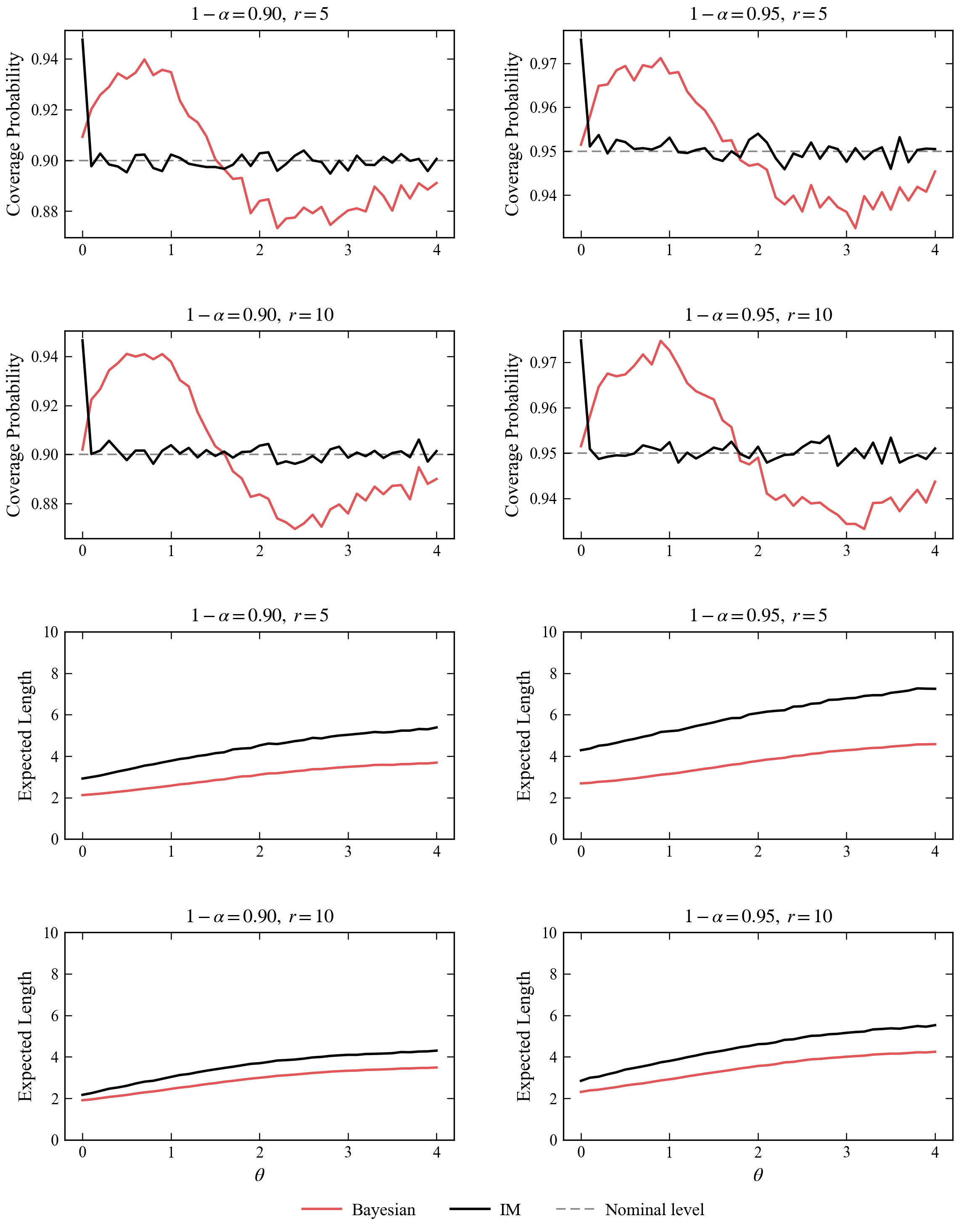}
\end{center}
\noindent\textbf{Figure \thefigure.} Coverage probabilities and expected lengths of the Bayesian and IM confidence intervals for \(\theta \in 0.0\ (0.1)\ 4.0\) when \(r = 5,\ 10\) and \(1 - \alpha = 0.90,\ 0.95\).\label{fig:1}

\clearpage
\refstepcounter{figure}
\begin{center}
\includegraphics[width=1.0\textwidth,height=0.95\textheight,keepaspectratio]{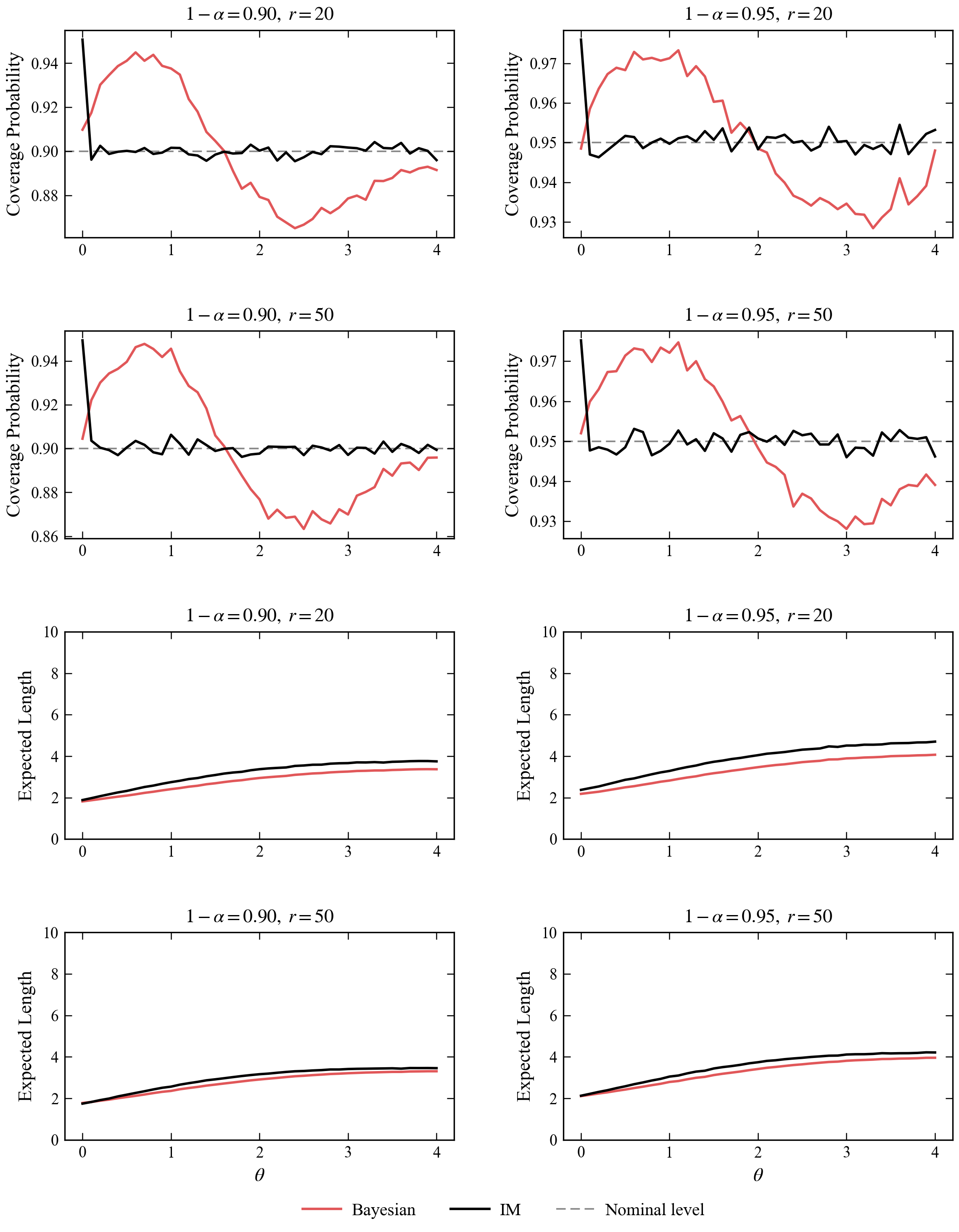}
\end{center}
\noindent\textbf{Figure \thefigure.} Coverage probabilities and expected lengths of the Bayesian and IM confidence intervals for \(\theta \in 0.0\ (0.1)\ 4.0\) when \(r = 20,\ 50\) and \(1 - \alpha = 0.90,\ 0.95\).\label{fig:2}

\clearpage
\refstepcounter{figure}
\begin{center}
\includegraphics[width=1.0\textwidth,height=0.95\textheight,keepaspectratio]{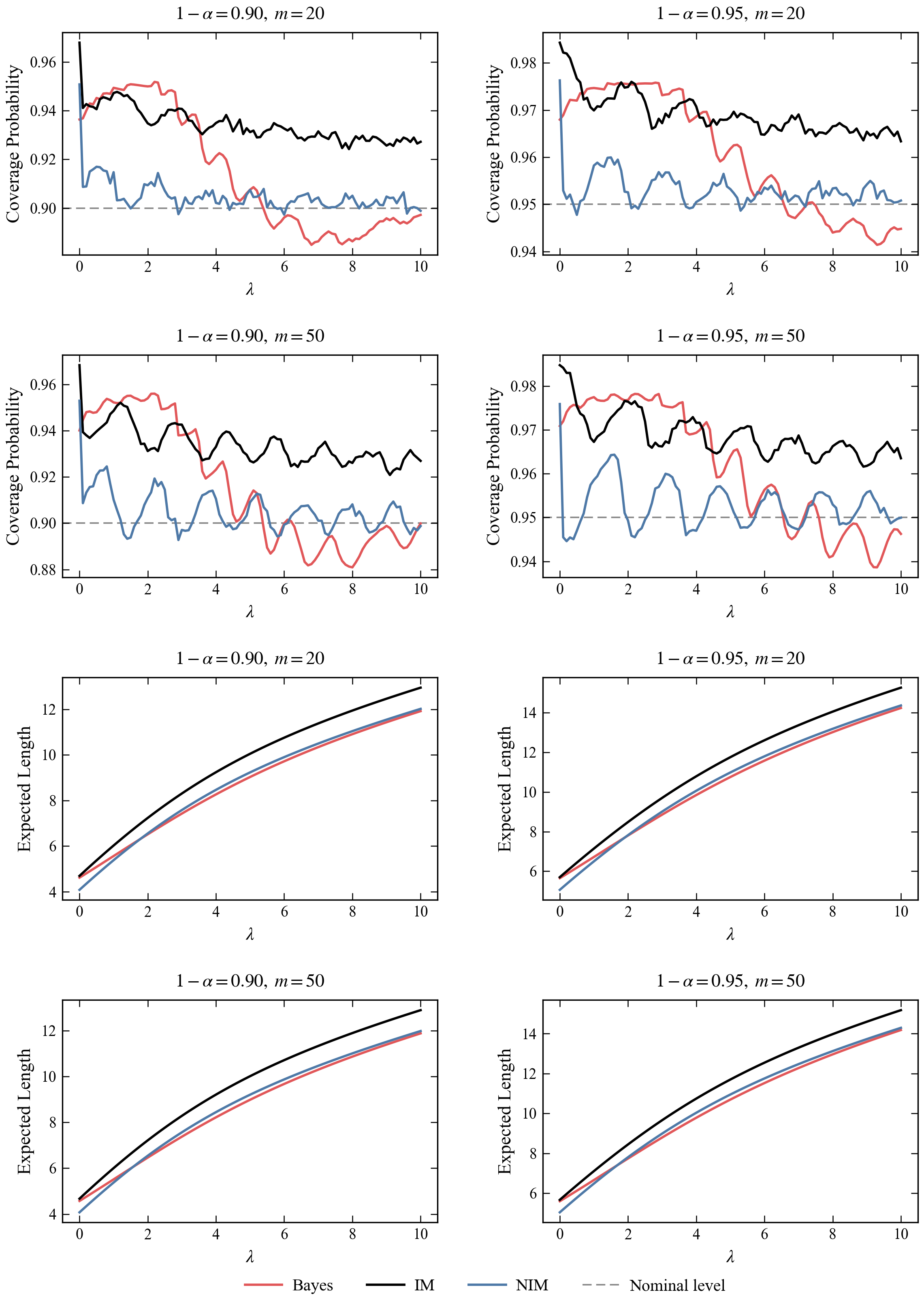}
\end{center}
\noindent\textbf{Figure \thefigure.} Coverage probabilities and expected lengths of the Bayesian, IM, and NIM confidence intervals for \(\lambda \in 0.0\ (0.1)\ 10.0\) when \(\varepsilon = 3.0\), \(m = 20,\ 50\) and \(1 - \alpha = 0.90,\ 0.95\).\label{fig:3}

\clearpage
\refstepcounter{figure}
\begin{center}
\includegraphics[width=1.0\textwidth,height=0.95\textheight,keepaspectratio]{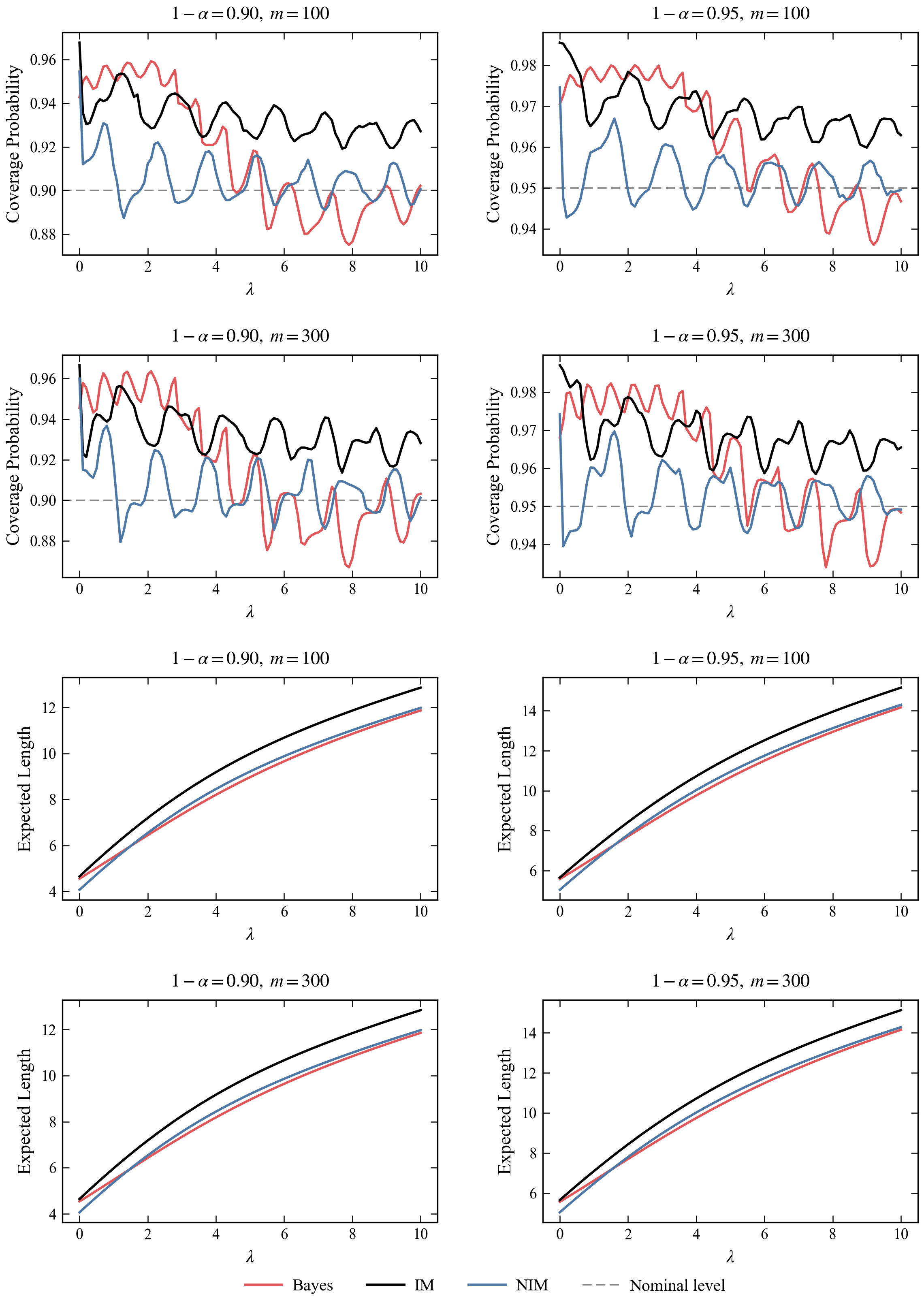}
\end{center}
\noindent\textbf{Figure \thefigure.} Coverage probabilities and expected lengths of the Bayesian, IM, and NIM confidence intervals for \(\lambda \in 0.0\ (0.1)\ 10.0\) when \(\varepsilon = 3.0\), \(m = 100,\ 300\) and \(1 - \alpha = 0.90,\ 0.95\).\label{fig:4}

\clearpage
\refstepcounter{figure}
\begin{center}
\includegraphics[width=0.95\textwidth,height=0.82\textheight,keepaspectratio]{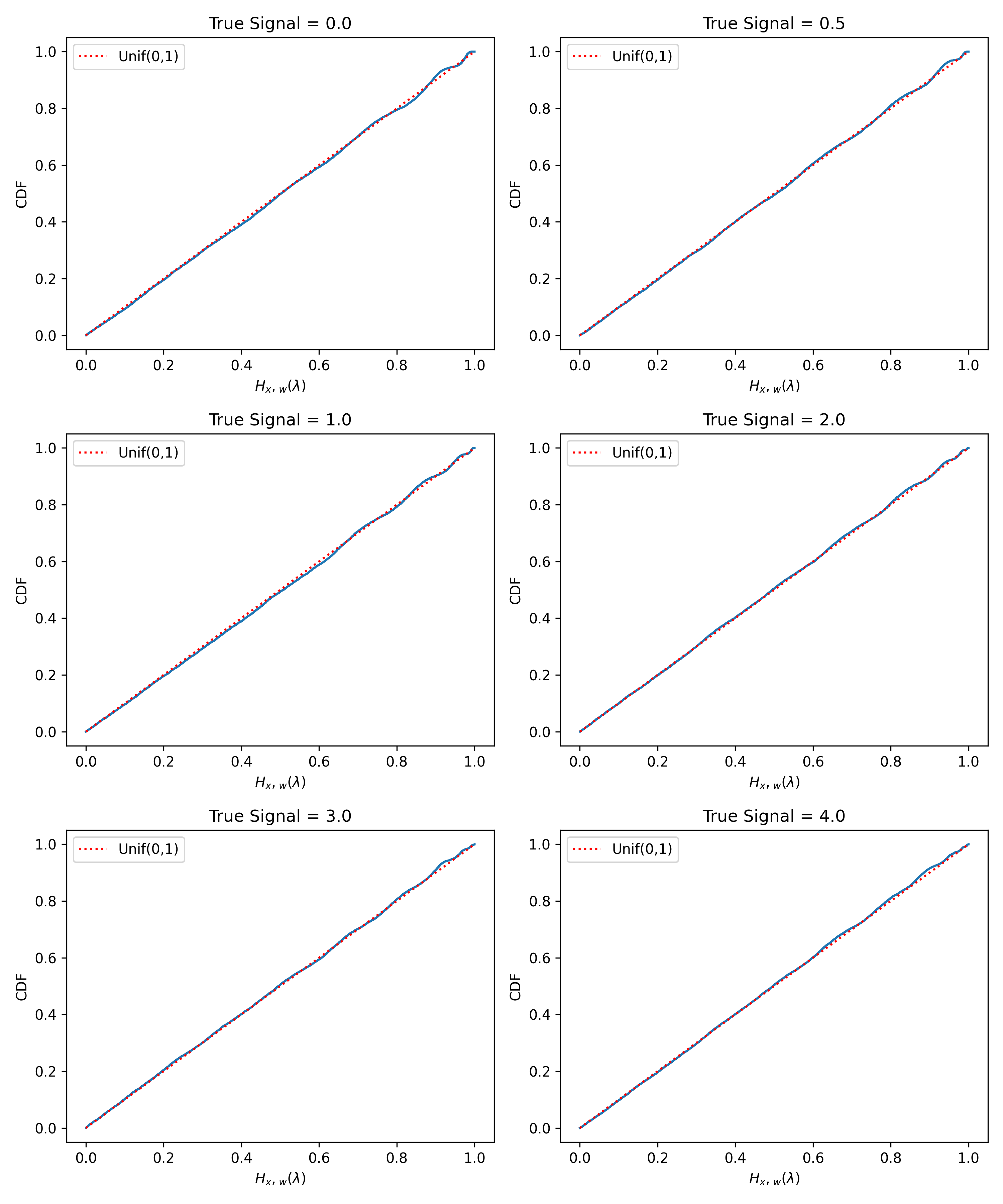}
\end{center}
\noindent\textbf{Figure \thefigure.} Empirical distribution functions of \(H_{x,w}(\lambda)\) (blue) compared with those of Unif(0,1) (red) based on the random samples \(X\sim Poisson(\theta)\) and \(W\sim Poisson(m\varepsilon)\), where \(\theta = \varepsilon + \lambda\), \(\varepsilon = 3.0\), \(\lambda = 0.0,\ 0.5,\ 1.0,\ 2.0,\ 3.0,\ 4.0\), and \(m = 20\).\label{fig:5}

\clearpage
\refstepcounter{figure}
\begin{center}
\includegraphics[width=0.95\textwidth,height=0.82\textheight,keepaspectratio]{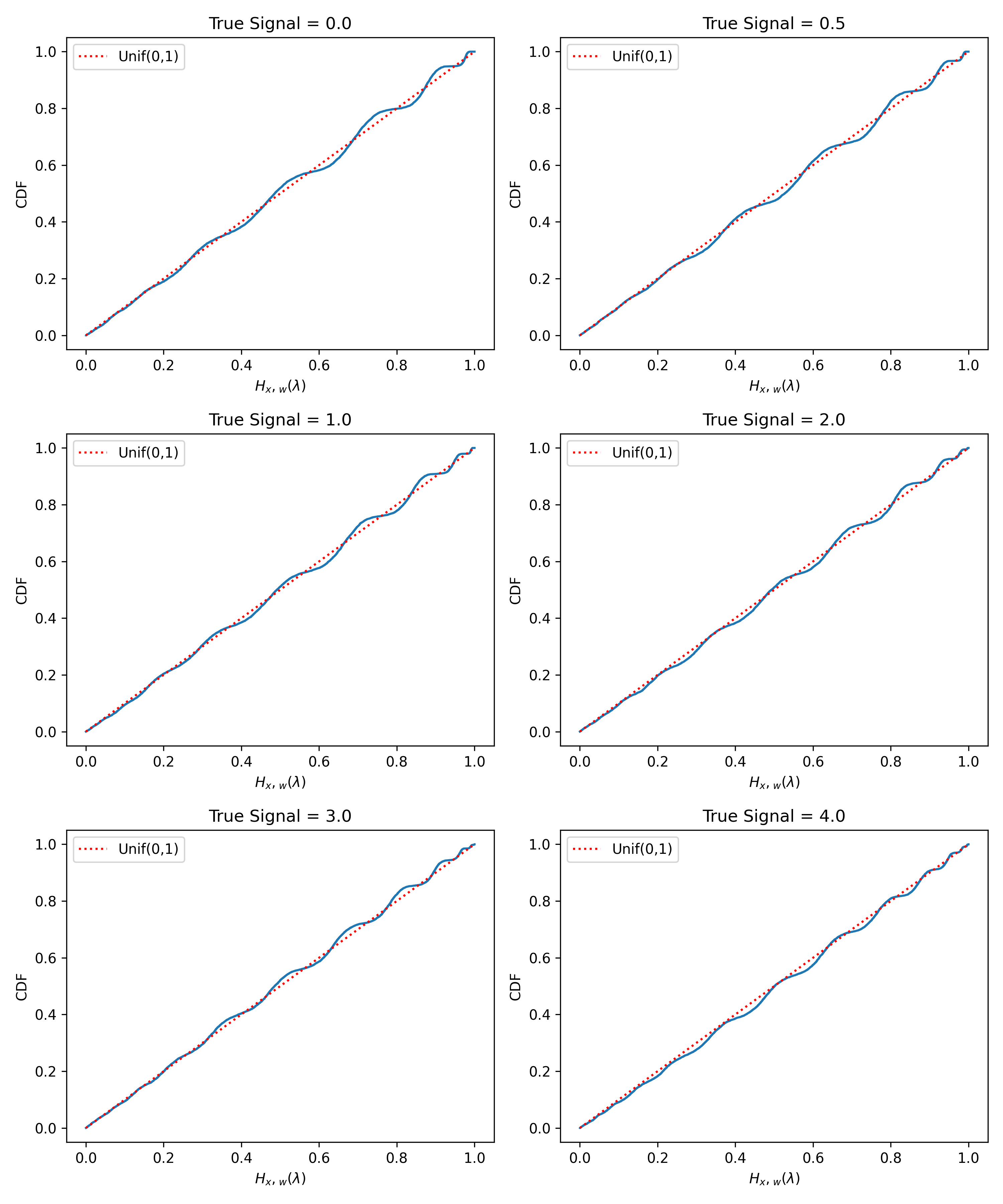}
\end{center}
\noindent\textbf{Figure \thefigure.} Empirical distribution functions of \(H_{x,w}(\lambda)\) (blue) compared with those of Unif(0,1) (red) based on the random samples \(X\sim Poisson(\theta)\) and \(W\sim Poisson(m\varepsilon)\), where \(\theta = \varepsilon + \lambda\), \(\varepsilon = 3.0\), \(\lambda = 0.0,\ 0.5,\ 1.0,\ 2.0,\ 3.0,\ 4.0\), and \(m = 50\).\label{fig:6}

\clearpage
\refstepcounter{figure}
\begin{center}
\includegraphics[width=0.95\textwidth,height=0.82\textheight,keepaspectratio]{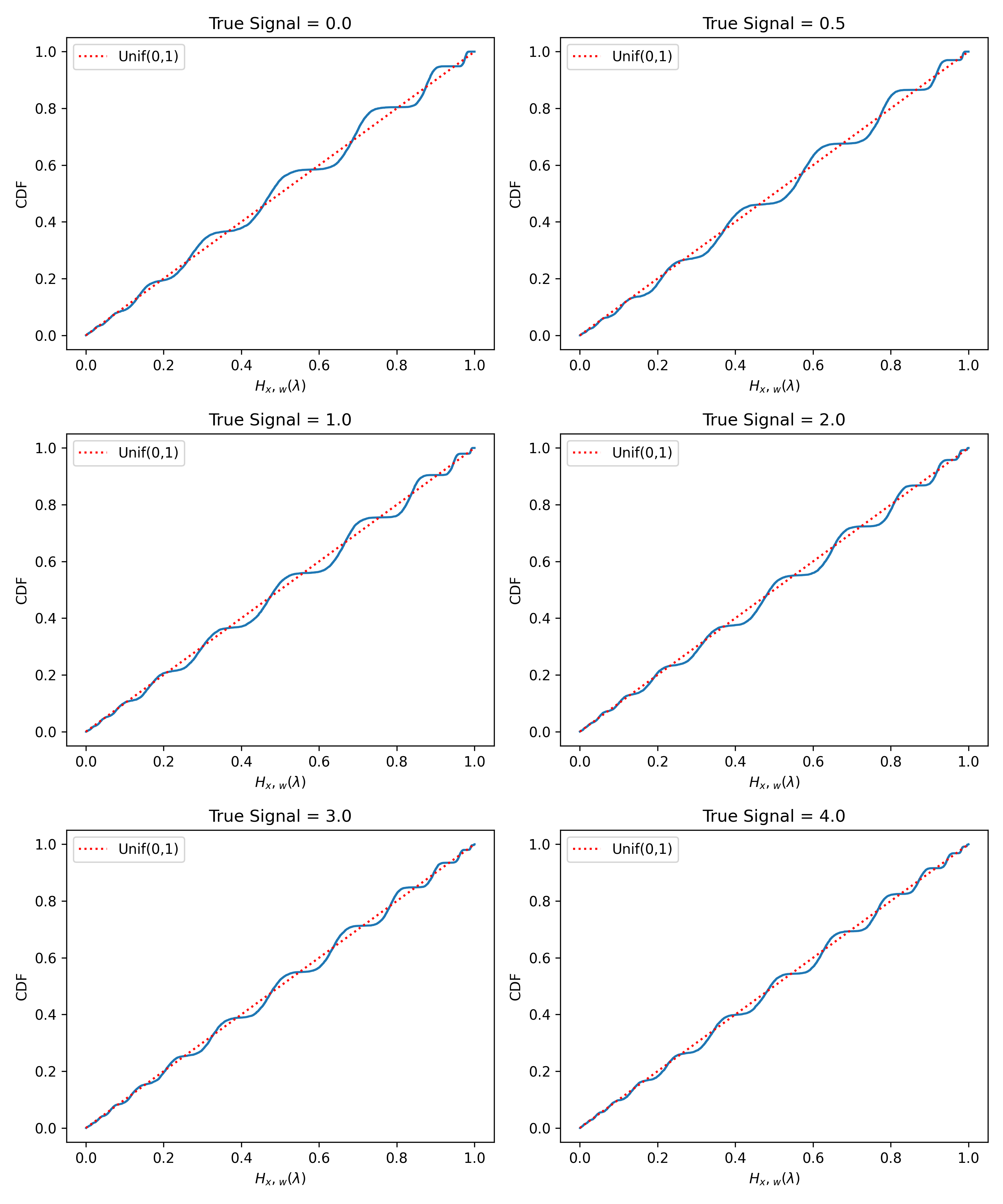}
\end{center}
\noindent\textbf{Figure \thefigure.} Empirical distribution functions of \(H_{x,w}(\lambda)\) (blue) compared with those of Unif(0,1) (red) based on the random samples \(X\sim Poisson(\theta)\) and \(W\sim Poisson(m\varepsilon)\), where \(\theta = \varepsilon + \lambda\), \(\varepsilon = 3.0\), \(\lambda = 0.0,\ 0.5,\ 1.0,\ 2.0,\ 3.0,\ 4.0\), and \(m = 100\).\label{fig:7}

\clearpage
\refstepcounter{figure}
\begin{center}
\includegraphics[width=0.95\textwidth,height=0.82\textheight,keepaspectratio]{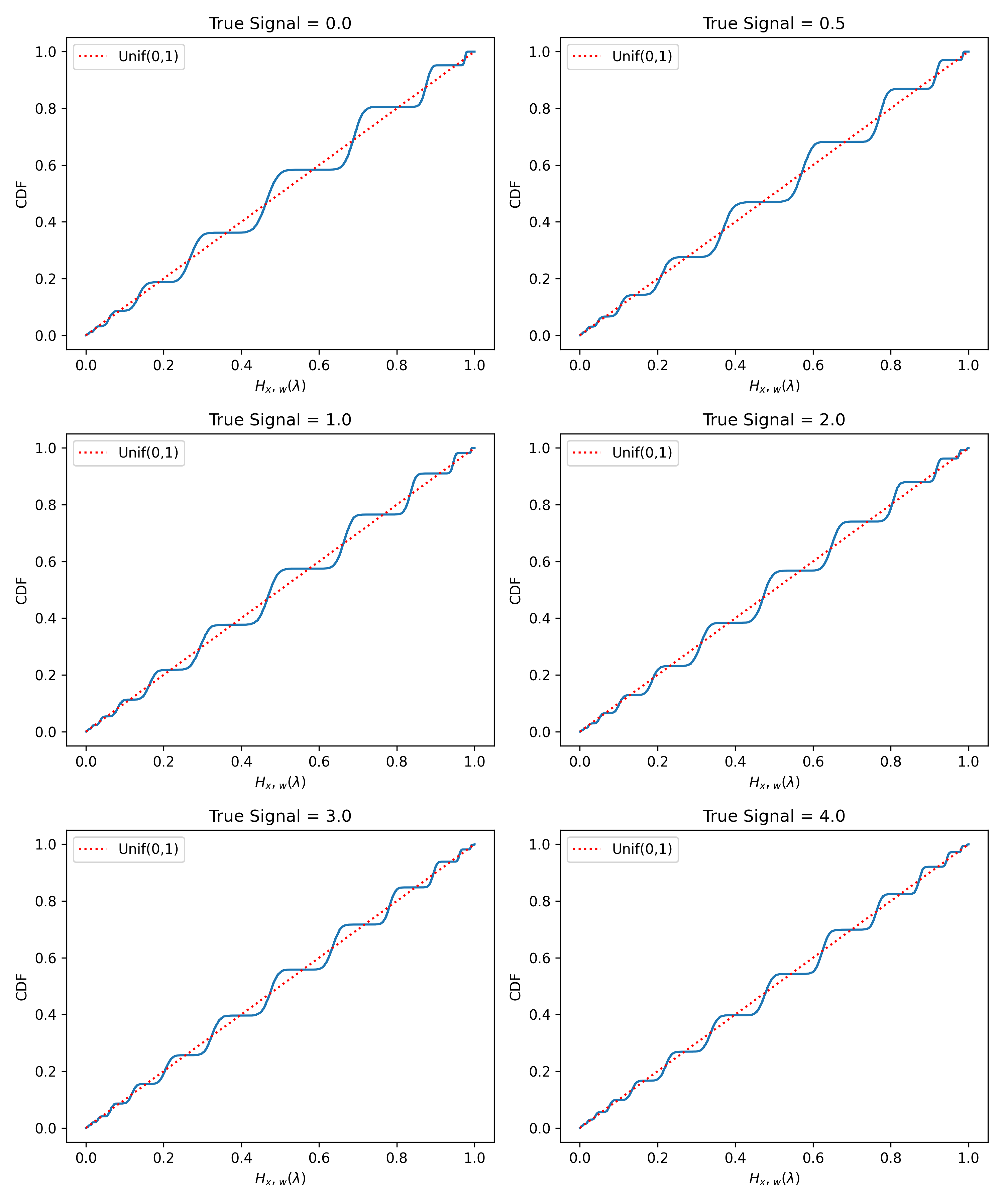}
\end{center}
\noindent\textbf{Figure \thefigure.} Empirical distribution functions of \(H_{x,w}(\lambda)\) (blue) compared with those of Unif(0,1) (red) based on the random samples \(X\sim Poisson(\theta)\) and \(W\sim Poisson(m\varepsilon)\), where \(\theta = \varepsilon + \lambda\), \(\varepsilon = 3.0\), \(\lambda = 0.0,\ 0.5,\ 1.0,\ 2.0,\ 3.0,\ 4.0\), and \(m = 300\).\label{fig:8}

\clearpage
\refstepcounter{figure}
\begin{center}
\includegraphics[width=0.95\textwidth,height=0.82\textheight,keepaspectratio]{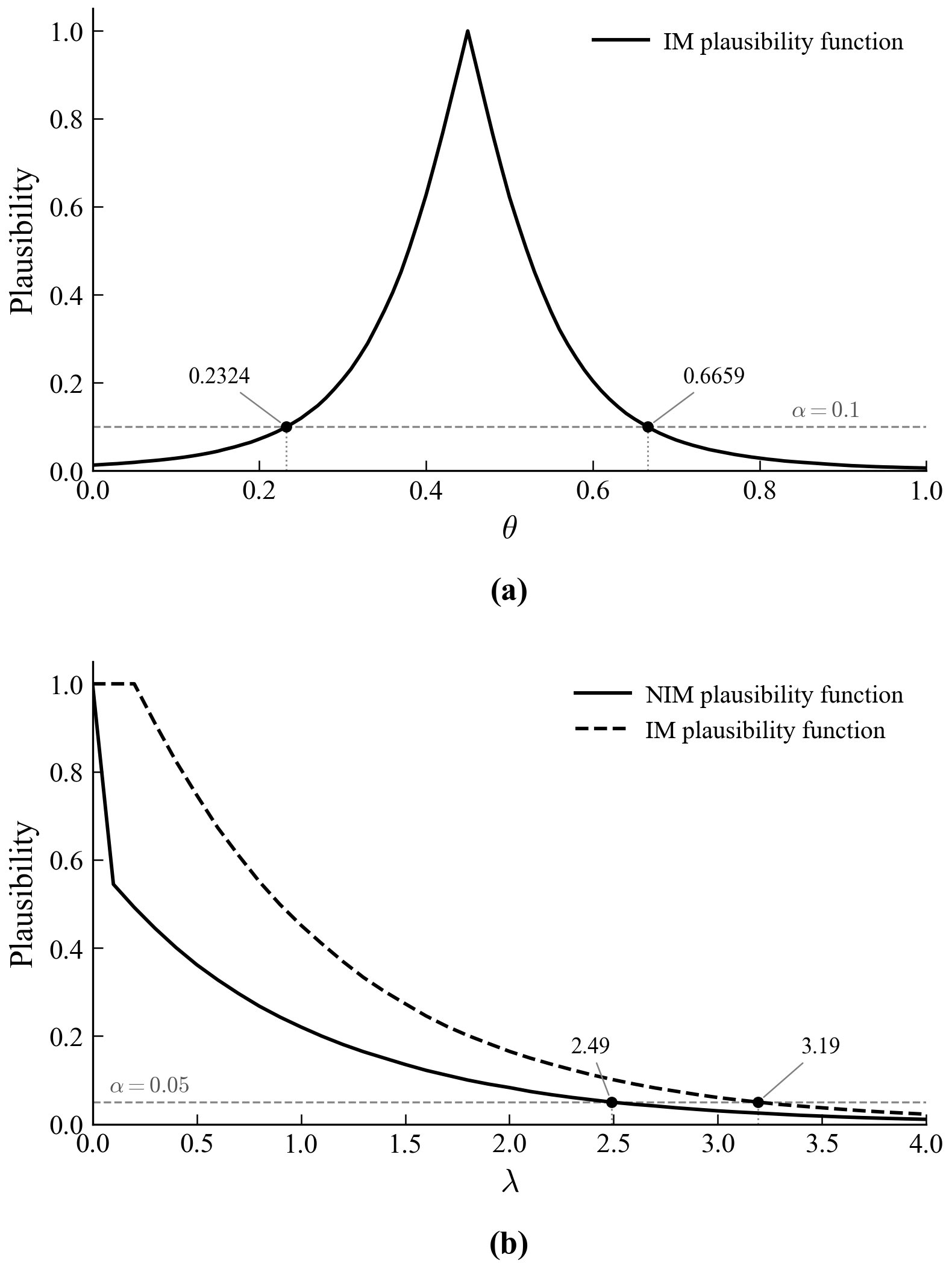}
\end{center}
\noindent\textbf{Figure \thefigure.} Plausibility functions for constrained normal and Poisson datasets. (a) \(X = 0.45\), \(W = 1.0\), and \(r = 10\); (b) \(X = 0\), \(W = 10\), and \(m = 20\).\label{fig:9}

\clearpage
\refstepcounter{table}\noindent\textbf{Table 1.} Levels 0.90 and 0.95 Bayesian and IM CIs and their interval widths under different combinations of observations \((X,W,r)\).\label{tab:1}

\begin{longtable}[]{@{}
  >{\centering\arraybackslash}p{(\linewidth - 10\tabcolsep) * \real{0.0852}}
  >{\centering\arraybackslash}p{(\linewidth - 10\tabcolsep) * \real{0.0682}}
  >{\centering\arraybackslash}p{(\linewidth - 10\tabcolsep) * \real{0.2116}}
  >{\centering\arraybackslash}p{(\linewidth - 10\tabcolsep) * \real{0.2117}}
  >{\centering\arraybackslash}p{(\linewidth - 10\tabcolsep) * \real{0.2116}}
  >{\centering\arraybackslash}p{(\linewidth - 10\tabcolsep) * \real{0.2117}}@{}}
\toprule\noalign{}
\multicolumn{2}{@{}>{\centering\arraybackslash}p{(\linewidth - 10\tabcolsep) * \real{0.1534} + 2\tabcolsep}}{%
\begin{minipage}[b]{\linewidth}\centering
Parameter
\end{minipage}} & \multicolumn{2}{>{\centering\arraybackslash}p{(\linewidth - 10\tabcolsep) * \real{0.4233} + 2\tabcolsep}}{%
\begin{minipage}[b]{\linewidth}\centering
Level 0.90 CIs and interval widths
\end{minipage}} & \multicolumn{2}{>{\centering\arraybackslash}p{(\linewidth - 10\tabcolsep) * \real{0.4233} + 2\tabcolsep}@{}}{%
\begin{minipage}[b]{\linewidth}\centering
Level 0.95 CIs and interval widths
\end{minipage}} \\
\midrule\noalign{}
\endhead
\bottomrule\noalign{}
\endlastfoot
\(W\) & \(r\) & Bayesian & IM & Bayesian & IM \\
0.01 & 5 & {[}0.3599, 0.5401{]} & {[}0.4442, 0.4558{]} & {[}0.3351, 0.5649{]} & {[}0.4415, 0.4585{]} \\
& & 0.1802 & \ul{0.0116} & 0.2298 & \ul{0.0171} \\
& 10 & {[}0.3927, 0.5073{]} & {[}0.4478, 0.4522{]} & {[}0.3795, 0.5205{]} & {[}0.4471, 0.4529{]} \\
& & 0.1146 & \ul{0.0043} & 0.1409 & \ul{0.0058} \\
& 20 & {[}0.4114, 0.4886{]} & {[}0.4491, 0.4509{]} & {[}0.4034, 0.4966{]} & {[}0.4488, 0.4512{]} \\
& & 0.0771 & \ul{0.0019} & 0.0933 & \ul{0.0024} \\
& 50 & {[}0.4263, 0.4737{]} & {[}0.4497, 0.4503{]} & {[}0.4216, 0.4784{]} & {[}0.4496, 0.4504{]} \\
& & 0.0474 & \ul{0.0007} & 0.0568 & \ul{0.0008} \\
0.10 & 5 & {[}0.1766, 0.7234{]} & {[}0.3919, 0.5078{]} & {[}0.1106, 0.7894{]} & {[}0.3643, 0.5354{]} \\
& & 0.5468 & \ul{0.1159} & 0.6788 & \ul{0.1711} \\
& 10 & {[}0.2691, 0.6309{]} & {[}0.4283, 0.4718{]} & {[}0.2278, 0.6722{]} & {[}0.4211, 0.4789{]} \\
& & 0.3619 & \ul{0.0434} & 0.4444 & \ul{0.0578} \\
& 20 & {[}0.3280, 0.5720{]} & {[}0.4406, 0.4595{]} & {[}0.3025, 0.5975{]} & {[}0.4381, 0.4619{]} \\
& & 0.2439 & \ul{0.0189} & 0.2950 & \ul{0.0238} \\
& 50 & {[}0.3751, 0.5249{]} & {[}0.4465, 0.4535{]} & {[}0.3602, 0.5398{]} & {[}0.4458, 0.4542{]} \\
& & 0.1499 & \ul{0.0070} & 0.1797 & \ul{0.0085} \\
0.50 & 5 & {[}0.0000, 0.9446{]} & {[}0.1590, 0.7385{]} & {[}0.0000, 1.1153{]} & {[}0.0235, 0.8769{]} \\
& & 0.9446 & \ul{0.5795} & 1.1153 & \ul{0.8533} \\
& 10 & {[}0.0836, 0.8164{]} & {[}0.3416, 0.5588{]} & {[}0.0212, 0.8788{]} & {[}0.3059, 0.5939{]} \\
& & 0.7328 & \ul{0.2172} & 0.8576 & \ul{0.2880} \\
& 20 & {[}0.1811, 0.7189{]} & {[}0.4028, 0.4971{]} & {[}0.1274, 0.7726{]} & {[}0.3907, 0.5094{]} \\
& & 0.5378 & \ul{0.0943} & 0.6452 & \ul{0.1187} \\
& 50 & {[}0.2824, 0.6176{]} & {[}0.4326, 0.4674{]} & {[}0.2492, 0.6508{]} & {[}0.4288, 0.4712{]} \\
& & 0.3352 & \ul{0.0348} & 0.4017 & \ul{0.0424} \\
1.00 & 5 & {[}0.0000, 1.1791{]} & {[}0.0000, 1.0286{]} & {[}0.0000, 1.4212{]} & {[}0.0000, 1.3117{]} \\
& & 1.1791 & \ul{1.0286} & 1.4212 & \ul{1.3117} \\
& 10 & {[}0.0000, 0.9042{]} & {[}0.2324, 0.6659{]} & {[}0.0000, 1.0419{]} & {[}0.1624, 0.7374{]} \\
& & 0.9042 & \ul{0.4335} & 1.0419 & \ul{0.5750} \\
& 20 & {[}0.0930, 0.8070{]} & {[}0.3556, 0.5441{]} & {[}0.0338, 0.8662{]} & {[}0.3313, 0.5689{]} \\
& & 0.7140 & \ul{0.1886} & 0.8324 & \ul{0.2376} \\
& 50 & {[}0.2138, 0.6862{]} & {[}0.4153, 0.4849{]} & {[}0.1675, 0.7325{]} & {[}0.4076, 0.4923{]} \\
& & 0.4724 & \ul{0.0696} & 0.5651 & \ul{0.0847} \\
5.00 & 5 & {[}0.0000, 2.2436{]} & {[}0.0000, 3.3612{]} & {[}0.0000, 2.7895{]} & {[}0.0000, 4.7088{]} \\
& & \ul{2.2436} & 3.3612 & \ul{2.7895} & 4.7088 \\
& 10 & {[}0.0000, 1.5644{]} & {[}0.0000, 1.5318{]} & {[}0.0000, 1.8661{]} & {[}0.0000, 1.8898{]} \\
& & 1.5644 & \ul{1.5318} & \ul{1.8661} & 1.8898 \\
& 20 & {[}0.0070, 1.1758{]} & {[}0.0000, 0.9216{]} & {[}0.0000, 1.3687{]} & {[}0.0000, 1.0454{]} \\
& & 1.1758 & \ul{0.9216} & 1.3687 & \ul{1.0454} \\
& 50 & {[}0.0123, 0.8877{]} & {[}0.2763, 0.6234{]} & {[}0.0000, 0.9933{]} & {[}0.2386, 0.6620{]} \\
& & 0.8755 & \ul{0.3472} & 0.9933 & \ul{0.4234} \\
10.00 & 5 & {[}0.0000, 3.0651{]} & {[}0.0000, 6.2209{]} & {[}0.0000, 3.8400{]} & {[}0.0000, 8.9531{]} \\
& & \ul{3.0651} & 6.2209 & \ul{3.8400} & 8.9531 \\
& 10 & {[}0.0000, 2.0818{]} & {[}0.0000, 2.6215{]} & {[}0.0000, 2.5055{]} & {[}0.0000, 3.3284{]} \\
& & \ul{2.0818} & 2.6215 & \ul{2.5055} & 3.3284 \\
& 20 & {[}0.0000, 1.5173{]} & {[}0.0000, 1.3945{]} & {[}0.0000, 1.7861{]} & {[}0.0000, 1.6341{]} \\
& & 1.5173 & \ul{1.3945} & 1.7861 & \ul{1.6341} \\
& 50 & {[}0.0000, 1.0755{]} & {[}0.1020, 0.7975{]} & {[}0.0000, 1.2384{]} & {[}0.0268, 0.8725{]} \\
& & 1.0755 & \ul{0.6955} & 1.2384 & \ul{0.8457} \\
\end{longtable}

Note that cases in which the interval length is the shortest appear in bold underlined.

\newpage

\setcounter{table}{1}
\refstepcounter{table}\noindent\textbf{Table 2.} Levels 0.90 and 0.95 Bayesian, IM and NIM CIs and their interval widths under different combinations of observations \((x,w,\ m)\).\label{tab:2}

\begin{longtable}[]{@{}
  >{\centering\arraybackslash}p{(\linewidth - 16\tabcolsep) * \real{0.0477}}
  >{\centering\arraybackslash}p{(\linewidth - 16\tabcolsep) * \real{0.0772}}
  >{\centering\arraybackslash}p{(\linewidth - 16\tabcolsep) * \real{0.0545}}
  >{\centering\arraybackslash}p{(\linewidth - 16\tabcolsep) * \real{0.1368}}
  >{\centering\arraybackslash}p{(\linewidth - 16\tabcolsep) * \real{0.1368}}
  >{\centering\arraybackslash}p{(\linewidth - 16\tabcolsep) * \real{0.1368}}
  >{\centering\arraybackslash}p{(\linewidth - 16\tabcolsep) * \real{0.1368}}
  >{\centering\arraybackslash}p{(\linewidth - 16\tabcolsep) * \real{0.1368}}
  >{\centering\arraybackslash}p{(\linewidth - 16\tabcolsep) * \real{0.1368}}@{}}
\toprule\noalign{}
\multicolumn{3}{@{}>{\centering\arraybackslash}p{(\linewidth - 16\tabcolsep) * \real{0.1794} + 4\tabcolsep}}{%
\begin{minipage}[b]{\linewidth}\centering
Parameter
\end{minipage}} & \multicolumn{3}{>{\centering\arraybackslash}p{(\linewidth - 16\tabcolsep) * \real{0.4103} + 4\tabcolsep}}{%
\begin{minipage}[b]{\linewidth}\centering
Level 0.90 CIs and interval widths
\end{minipage}} & \multicolumn{3}{>{\centering\arraybackslash}p{(\linewidth - 16\tabcolsep) * \real{0.4103} + 4\tabcolsep}@{}}{%
\begin{minipage}[b]{\linewidth}\centering
Level 0.95 CIs and interval widths
\end{minipage}} \\
\midrule\noalign{}
\endhead
\bottomrule\noalign{}
\endlastfoot
\(x\) & \(m\) & \(w\) & Bayesian & IM & NIM & Bayesian & IM & NIM \\
0 & 20 & 10 & {[}0.00, 2.30{]} & {[}0.00, 2.50{]} & {[}0.00, 1.78{]} & {[}0.00, 3.00{]} & {[}0.00, 3.19{]} & {[}0.00, 2.49{]} \\
& & & 2.30 & 2.50 & \ul{1.78} & 3.00 & 3.19 & \ul{2.49} \\
& & 20 & {[}0.00, 2.30{]} & {[}0.00, 2.02{]} & {[}0.00, 1.29{]} & {[}0.00, 3.00{]} & {[}0.00, 2.72{]} & {[}0.00, 1.96{]} \\
& & & 2.30 & 2.02 & \ul{1.29} & 3.00 & 2.72 & \ul{1.96} \\
& & 30 & {[}0.00, 2.30{]} & {[}0.00, 1.53{]} & {[}0.00, 0.80{]} & {[}0.00, 3.00{]} & {[}0.00, 2.23{]} & {[}0.00, 1.47{]} \\
& & & 2.30 & 1.53 & \ul{0.80} & 3.00 & 2.23 & \ul{1.47} \\
& & 40 & {[}0.00, 2.30{]} & {[}0.00, 1.05{]} & {[}0.00, 0.31{]} & {[}0.00, 3.00{]} & {[}0.00, 1.73{]} & {[}0.00, 0.98{]} \\
& & & 2.30 & 1.05 & \ul{0.31} & 3.00 & 1.73 & \ul{0.98} \\
0 & 50 & 10 & {[}0.00, 2.30{]} & {[}0.00, 2.80{]} & {[}0.00, 2.08{]} & {[}0.00, 3.00{]} & {[}0.00, 3.48{]} & {[}0.00, 2.73{]} \\
& & & 2.30 & 2.80 & \ul{2.08} & 3.00 & 3.48 & \ul{2.73} \\
& & 20 & {[}0.00, 2.30{]} & {[}0.00, 2.60{]} & {[}0.00, 1.89{]} & {[}0.00, 3.00{]} & {[}0.00, 3.30{]} & {[}0.00, 2.53{]} \\
& & & 2.30 & 2.60 & \ul{1.89} & 3.00 & 3.30 & \ul{2.53} \\
& & 30 & {[}0.00, 2.30{]} & {[}0.00, 2.41{]} & {[}0.00, 1.69{]} & {[}0.00, 3.00{]} & {[}0.00, 3.10{]} & {[}0.00, 2.34{]} \\
& & & 2.30 & 2.41 & \ul{1.69} & 3.00 & 3.10 & \ul{2.34} \\
& & 40 & {[}0.00, 2.30{]} & {[}0.00, 2.21{]} & {[}0.00, 1.49{]} & {[}0.00, 3.00{]} & {[}0.00, 2.90{]} & {[}0.00, 2.14{]} \\
& & & 2.30 & 2.21 & \ul{1.49} & 3.00 & 2.90 & \ul{2.14} \\
0 & 100 & 10 & {[}0.00, 2.30{]} & {[}0.00, 2.89{]} & {[}0.00, 2.20{]} & {[}0.00, 3.00{]} & {[}0.00, 3.58{]} & {[}0.00, 2.83{]} \\
& & & 2.30 & 2.89 & \ul{2.20} & 3.00 & 3.58 & \ul{2.83} \\
& & 20 & {[}0.00, 2.30{]} & {[}0.00, 2.79{]} & {[}0.00, 2.10{]} & {[}0.00, 3.00{]} & {[}0.00, 3.50{]} & {[}0.00, 2.73{]} \\
& & & 2.30 & 2.79 & \ul{2.10} & 3.00 & 3.50 & \ul{2.73} \\
& & 30 & {[}0.00, 2.30{]} & {[}0.00, 2.70{]} & {[}0.00, 2.00{]} & {[}0.00, 3.00{]} & {[}0.00, 3.39{]} & {[}0.00, 2.63{]} \\
& & & 2.30 & 2.70 & \ul{2.00} & 3.00 & 3.39 & \ul{2.63} \\
& & 40 & {[}0.00, 2.30{]} & {[}0.00, 2.59{]} & {[}0.00, 1.89{]} & {[}0.00, 3.00{]} & {[}0.00, 3.30{]} & {[}0.00, 2.53{]} \\
& & & 2.30 & 2.59 & \ul{1.89} & 3.00 & 3.30 & \ul{2.53} \\
0 & 300 & 10 & {[}0.00, 2.30{]} & {[}0.00, 2.96{]} & {[}0.00, 2.26{]} & {[}0.00, 3.00{]} & {[}0.00, 3.66{]} & {[}0.00, 2.90{]} \\
& & & 2.30 & 2.96 & \ul{2.26} & 3.00 & 3.66 & \ul{2.90} \\
& & 20 & {[}0.00, 2.30{]} & {[}0.00, 2.93{]} & {[}0.00, 2.23{]} & {[}0.00, 3.00{]} & {[}0.00, 3.64{]} & {[}0.00, 2.87{]} \\
& & & 2.30 & 2.93 & \ul{2.23} & 3.00 & 3.64 & \ul{2.87} \\
& & 30 & {[}0.00, 2.30{]} & {[}0.00, 2.90{]} & {[}0.00, 2.20{]} & {[}0.00, 3.00{]} & {[}0.00, 3.58{]} & {[}0.00, 2.83{]} \\
& & & 2.30 & 2.90 & \ul{2.20} & 3.00 & 3.58 & \ul{2.83} \\
& & 40 & {[}0.00, 2.30{]} & {[}0.00, 2.86{]} & {[}0.00, 2.17{]} & {[}0.00, 3.00{]} & {[}0.00, 3.55{]} & {[}0.00, 2.80{]} \\
& & & 2.30 & 2.86 & \ul{2.17} & 3.00 & 3.55 & \ul{2.80} \\
1 & 20 & 10 & {[}0.00, 3.51{]} & {[}0.00, 4.25{]} & {[}0.00, 3.62{]} & {[}0.00, 4.36{]} & {[}0.00, 5.07{]} & {[}0.00, 4.38{]} \\
& & & \ul{3.51} & 4.25 & 3.62 & \ul{4.36} & 5.07 & 4.38 \\
& & 20 & {[}0.00, 3.29{]} & {[}0.00, 3.77{]} & {[}0.00, 3.13{]} & {[}0.00, 4.13{]} & {[}0.00, 4.59{]} & {[}0.00, 3.90{]} \\
& & & 3.29 & 3.77 & \ul{3.13} & 4.13 & 4.59 & \ul{3.90} \\
& & 30 & {[}0.00, 3.14{]} & {[}0.00, 3.28{]} & {[}0.00, 2.63{]} & {[}0.00, 3.97{]} & {[}0.00, 4.12{]} & {[}0.00, 3.42{]} \\
& & & 3.14 & 3.28 & \ul{2.63} & 3.97 & 4.12 & \ul{3.42} \\
& & 40 & {[}0.00, 3.02{]} & {[}0.00, 2.78{]} & {[}0.00, 2.14{]} & {[}0.00, 3.85{]} & {[}0.00, 3.61{]} & {[}0.00, 2.93{]} \\
& & & 3.02 & 2.78 & \ul{2.14} & 3.85 & 3.61 & \ul{2.93} \\
1 & 50 & 10 & {[}0.00, 3.70{]} & {[}0.00, 4.55{]} & {[}0.00, 3.92{]} & {[}0.00, 4.56{]} & {[}0.00, 5.37{]} & {[}0.00, 4.68{]} \\
& & & \ul{3.70} & 4.55 & 3.92 & \ul{4.56} & 5.37 & 4.68 \\
& & 20 & {[}0.00, 3.57{]} & {[}0.00, 4.35{]} & {[}0.00, 3.72{]} & {[}0.00, 4.42{]} & {[}0.00, 5.17{]} & {[}0.00, 4.48{]} \\
& & & \ul{3.57} & 4.35 & 3.72 & \ul{4.42} & 5.17 & 4.48 \\
& & 30 & {[}0.00, 3.45{]} & {[}0.00, 4.15{]} & {[}0.00, 3.53{]} & {[}0.00, 4.30{]} & {[}0.00, 4.98{]} & {[}0.00, 4.29{]} \\
& & & \ul{3.45} & 4.15 & 3.53 & 4.30 & 4.98 & \ul{4.29} \\
& & 40 & {[}0.00, 3.36{]} & {[}0.00, 3.95{]} & {[}0.00, 3.33{]} & {[}0.00, 4.21{]} & {[}0.00, 4.79{]} & {[}0.00, 4.09{]} \\
& & & 3.36 & 3.95 & \ul{3.33} & 4.21 & 4.79 & \ul{4.09} \\
1 & 100 & 10 & {[}0.00, 3.79{]} & {[}0.00, 4.64{]} & {[}0.00, 4.02{]} & {[}0.00, 4.64{]} & {[}0.00, 5.48{]} & {[}0.00, 4.79{]} \\
& & & \ul{3.79} & 4.64 & 4.02 & \ul{4.64} & 5.48 & 4.79 \\
& & 20 & {[}0.00, 3.71{]} & {[}0.00, 4.55{]} & {[}0.00, 3.92{]} & {[}0.00, 4.56{]} & {[}0.00, 5.38{]} & {[}0.00, 4.69{]} \\
& & & \ul{3.71} & 4.55 & 3.92 & \ul{4.56} & 5.38 & 4.69 \\
& & 30 & {[}0.00, 3.63{]} & {[}0.00, 4.45{]} & {[}0.00, 3.82{]} & {[}0.00, 4.49{]} & {[}0.00, 5.28{]} & {[}0.00, 4.59{]} \\
& & & \ul{3.63} & 4.45 & 3.82 & \ul{4.49} & 5.28 & 4.59 \\
& & 40 & {[}0.00, 3.57{]} & {[}0.00, 4.35{]} & {[}0.00, 3.72{]} & {[}0.00, 4.42{]} & {[}0.00, 5.16{]} & {[}0.00, 4.49{]} \\
& & & \ul{3.57} & 4.35 & 3.72 & \ul{4.42} & 5.16 & 4.49 \\
1 & 300 & 10 & {[}0.05, 3.88{]} & {[}0.01, 4.71{]} & {[}0.07, 4.08{]} & {[}0.01, 4.73{]} & {[}0.00, 5.54{]} & {[}0.17, 4.86{]} \\
& & & \ul{3.83} & 4.71 & 4.01 & \ul{4.72} & 5.54 & 4.84 \\
& & 20 & {[}0.02, 3.84{]} & {[}0.00, 4.69{]} & {[}0.03, 4.05{]} & {[}0.00, 4.68{]} & {[}0.00, 5.50{]} & {[}0.00, 4.83{]} \\
& & & \ul{3.82} & 4.69 & 4.02 & \ul{4.68} & 5.50 & 4.83 \\
& & 30 & {[}0.00, 3.79{]} & {[}0.00, 4.64{]} & {[}0.00, 4.02{]} & {[}0.00, 4.65{]} & {[}0.00, 5.46{]} & {[}0.00, 4.79{]} \\
& & & \ul{3.79} & 4.64 & 4.02 & \ul{4.65} & 5.46 & 4.79 \\
& & 40 & {[}0.00, 3.76{]} & {[}0.00, 4.61{]} & {[}0.00, 3.99{]} & {[}0.00, 4.62{]} & {[}0.00, 5.44{]} & {[}0.00, 4.76{]} \\
& & & \ul{3.76} & 4.61 & 3.99 & \ul{4.62} & 5.44 & 4.76 \\
\end{longtable}

Note that cases in which the interval length is the shortest appear in bold underlined.

\end{document}